\documentclass[twocolumn]{IEEEtran}

\usepackage[english]{babel}

\usepackage{amsthm}
\usepackage{amsmath}
\usepackage{amsfonts}
\usepackage{dsfont}
\usepackage{mathrsfs}
\usepackage{pifont}
\usepackage{amssymb}
\usepackage{verbatim}
\usepackage{upgreek}
\usepackage{color}
\usepackage[final]{graphicx}
\usepackage{algorithmic}
\usepackage{algorithm}
\usepackage{epsfig}
\usepackage{eucal}
\usepackage{cite}
\usepackage{subfigure}
\usepackage{stmaryrd}

\makeatletter
\def\maketag@@@#1{\hbox{\m@th\normalfont\normalsize#1}}
\makeatother

\newcommand{\ba}{\mbox{\boldmath{$a$}}}

\newcommand{\bbm}{\begin{bmatrix}}
\newcommand{\ebm}{\end{bmatrix}}

\newcommand{\bit}{\begin{itemize}}
\newcommand{\eit}{\end{itemize}}

\newcommand{\ben}{\begin{enumerate}}
\newcommand{\een}{\end{enumerate}}

\newcommand{\bdesc}{\begin{description}}
\newcommand{\edesc}{\end{description}}

\newcommand{\bea}{\begin{array}}
\newcommand{\eea}{\end{array}}

\newcommand{\tr}{\mbox{\rm Tr}\, }

\newcommand{\beqa}{\begin{eqnarray}}
\newcommand{\eeqa}{\end{eqnarray}}

\newcommand{\ds}{\displaystyle}

\newcommand{\Comment}[1]{}

\newtheorem{prop}{Proposition}

\def\R{{\mathds R}}

\def\C{{\mathds C}}

\def\cC{\mbox{$\CMcal C$}}

\def\cF{\mbox{$\mathcal F$}}
\def\cG{\mbox{$\mathcal G$}}

\def\cL{\mbox{$\mathcal L$}}
\def\cN{\mbox{$\CMcal N$}}
\def\cI{\mbox{$\CMcal I$}}

\def\cP{\mbox{$\mathcal P$}}

\newcommand{\be}{\begin{equation}}
\newcommand{\ee}{\end{equation}}

\newcommand{\bzero}{{\mbox{\boldmath $0$}}}

\newcommand{\boa}{{\mbox{\boldmath $a$}}}

\newcommand{\bn}{{\mbox{\boldmath $n$}}}
\newcommand{\bm}{{\mbox{\boldmath $m$}}}

\newcommand{\bv}{{\mbox{\boldmath $v$}}}
\newcommand{\bx}{{\mbox{\boldmath $x$}}}

\newcommand{\bq}{{\mbox{\boldmath $q$}}}
\newcommand{\bz}{{\mbox{\boldmath $z$}}}

\newcommand{\bI}{{\mbox{\boldmath $I$}}}
\newcommand{\bJ}{{\mbox{\boldmath $J$}}}

\newcommand{\bM}{{\mbox{\boldmath $M$}}}

\newcommand{\bP}{{\mbox{\boldmath $P$}}}

\newcommand{\bS}{{\mbox{\boldmath $S$}}}

\newcommand{\bZ}{{\mbox{\boldmath $Z$}}}

\DeclareMathOperator*{\argmax}{arg\,max}

\newcommand{\btheta}{{\mbox{\boldmath $\theta$}}}

\newcommand{\dmax}{\begin{displaystyle}\max\end{displaystyle}}

\newcommand{\test}{\mbox{$
\begin{array}{c}
\stackrel{ \stackrel{\textstyle H_1}{\textstyle >} }{
\stackrel{\textstyle <}{\textstyle H_0} }
\end{array}
$}}

\begin{document}

\title{A Unifying Framework for Adaptive Radar Detection in the Presence of Multiple Alternative Hypotheses}

\author{Pia Addabbo, \emph{Senior Member, IEEE}, Sudan Han, Filippo Biondi, \emph{Member, IEEE}, Gaetano Giunta, \emph{Senior Member, IEEE}, 
	and Danilo Orlando, \emph{Senior Member, IEEE}
	\thanks{P. Addabbo is with Universit\`a degli studi ``Giustino Fortunato'', Benevento, Italy. E-mail: {\tt 
			p.addabbo@unifortunato.eu}.}
	\thanks{S. Han is with Defense Innovation Institute, Beijing, China E-mail: {\tt xiaoxiaosu0626@163.com}.}
	\thanks{F. Biondi is with Italian Ministry of Defence. Email: {\tt biopippo@gmail.com}.}
	\thanks{G. Giunta is with the Department of Engineering, University of Roma Tre, 00146 Rome, Italy. E-mail: {\tt gaetano.giunta@uniroma3.it}.}
	\thanks{D. Orlando is with the Faculty of Engineering, Universit\`a degli Studi ``Niccol\`o Cusano'', 00166 Roma, Italy. 
		E-mail: {\tt danilo.orlando@unicusano.it}.}
}

\maketitle

\begin{abstract}
In this paper, we develop a new elegant framework relying on the Kullback-Leibler Information Criterion to address the 
design of one-stage adaptive detection architectures for multiple hypothesis testing problems. Specifically, at the 
design stage, we assume that several alternative hypotheses may be in force and that only one null hypothesis exists. 
Then, starting from the case where all the parameters are known and proceeding until the case where the adaptivity with
respect to the entire parameter set is required, we come up with decision schemes for multiple alternative hypotheses 
consisting of the sum between the compressed log-likelihood ratio based upon the available data and a penalty term 
accounting for the number of unknown parameters. The latter rises from suitable approximations of the Kullback-Leibler
Divergence between the true and a candidate probability density function. Interestingly, under specific constraints, the 
proposed decision schemes can share the constant false alarm rate property by virtue of the Invariance Principle.
Finally, we show the effectiveness of the proposed framework through the application to examples of practical value in 
the context of radar detection also in comparison with two-stage competitors. This analysis highlights that the 
architectures devised within the proposed framework represent an effective means to deal with detection problems where 
the uncertainty on some parameters leads to multiple alternative hypotheses.
\end{abstract}

\begin{IEEEkeywords}
Adaptive Radar Detection, Constant False Alarm Rate, Generalized Likelihood Ratio Test, Kullback-Leibler Information Criterion, Model Order Selection, Multiple Hypothesis Testing, Nuisance Parameters, Radar, Statistical Invariance.
\end{IEEEkeywords}

\section{Introduction}
Nowadays, modern radar systems incorporate sophisticated signal processing algorithms which take advantage of the computational power made available by recent advances in technology. This growth in complexity is dictated by the fact these systems have to face with more and more challenging scenarios where conventional algorithms might fail or exhibit poor performance. For instance, in target-rich environments, structured echoes contaminate data used to estimate the spectral properties of the interference (also known as training or secondary data) leading to a dramatic attenuation of the signal of interest components and, hence, to a nonnegligible number of missed detections \cite{bergin2002gmti}. In such case, radar system should be endowed with signal processing schemes capable of detecting and suppressing the outliers in order to make the training data set homogeneous \cite{wicksOutliers, AdveOuliers, JiangOutliers, RangaswamyOutliers}.
Another important example concerns high resolution radars, which can resolve a target into a number of different scattering centers depending on the radar bandwidth and the range extent of the target \cite{ScheerMelvin,PiaIET}. The classic approach to the detection of range-spread targets consists in processing one range bin at a time despite the fact that contiguous cells contain target energy. As a consequence, classic detection algorithms do not collect as much energy as possible to increase the Signal-to-Interference-plus-Noise Ratio (SINR).
To overcome this drawback, architectures capable of detecting distributed targets by exploiting a preassigned number of contiguous range bins have been developed \cite{WLiuRao,JunLiu02,Gerlach,DD,GLRT-based}.
These energy issues also hold for multiple point-like targets. In fact, detection algorithms which can take advantage of the total energy associated with each point-like target are highly desirable. 

In the open literature, the existing examples concerning the detection of either multiple point-like or range-spread targets share the assumption that the number of scatterers (or at least an upper bound on it) is known and are based upon the Maximum Likelihood Approach (MLA) \cite{BR_multipleTargets,BOR_multipleTargets}. 
However, in scenarios of practical value, such a priori information is not often available especially when the radar system is operating in search mode.
Moreover, the problem of jointly detecting multiple point-like targets is very difficult since target positions and, more importantly, target number are unknown parameters that must be estimated. Thus, conversely to the conventional detection problems that comprise two hypotheses, namely the noise-only (or null) and the signal-plus-noise (or alternative) hypothesis, this lack of a priori information naturally leads to multiple alternative hypotheses with the consequence that the radar engineer has to face with composite multiple hypothesis tests. 

Besides target-dense environments, another operating situation leading to multiple hypothesis tests is related 
to possible electronic attacks by adversary forces (jammers). 
These attacks comprise active techniques aimed at protecting a platform from being detected and tracked by the radar \cite{ScheerMelvin} through two approaches: masking and deception.
More precisely, noncoherent jammers or Noise-Like Jammers (NLJs) attempt to mask targets generating nondeceptive interference which blends into the thermal noise of the radar receiver degrading the radar sensitivity due to an increase of the Constant False Alarm Rate (CFAR) threshold \cite{ScheerMelvin,EW101,FarinaSkolnik}. On the other hand, the Coherent Jammers (CJs) illuminate the victim radar by means of low duty-cycle signals with specific parameters that, when estimated by the radar processor, force the latter to allocate resources to handle false targets. In fact, CJs are equipped with electronics apparatuses capable of receiving, modifying, amplifying, and retransmitting the radar's own signal to create false targets with radar's range, Doppler, and angle far away from the true position of the platform under protection \cite{FarinaSkolnik,ScheerMelvin}.

A possible way to react to this kind of interference relies on the use of decision schemes devised by modifying the 
conventional detection problem with additional hypotheses associated with the presence of such threats 
\cite{FarinaECM,carotenutoNLJ}. 
In \cite{FarinaECM}, adaptive detection and discrimination between useful signals and CJs in the presence of thermal 
noise, clutter, and possible NLJs is addressed by considering an additional hypothesis under which data contain the 
CJs only. In addition, the latter is assumed to lay on the orthogonal complement of the subspace spanned by the 
nominal steering vector (after whitening by the true covariance matrix of the composite disturbance). 
The resulting multiple hypothesis test is solved resorting to an approach based upon a generalized Neyman-Pearson 
criterion \cite{GeneralizedNP}.
However, from a computational point of view, detection threshold setting might require an onerous load and, 
more importantly, such solution is effective when the multiple hypotheses are not nested. As a matter of fact,
in the presence of nested hypotheses, the MLA and, hence, the Generalized Likelihood Ratio Test (GLRT),
may fail because the likelihood function monotonically increases with the hypothesis order (or model order). 
As a consequence, the MLA experiences a natural inclination to overestimate the hypothesis order. 
An alternative approach consists in looking for detection schemes that incorporate the expedients of the 
so-called Model Order Selection (MOS) rules \cite{Anderson,Stoica1,BICneath,BhansaliGIC}, 
which include the parameter number diversity to moderate the overestimation attitude of the MLA. In \cite{carotenutoNLJ}, 
the authors follow the last approach to conceive two-stage detection architectures for multiple NLJs whose number is
unknown. Specifically, the first stage exploits the MOS rules to provide an estimate of the number of NLJs, 
whereas the second stage consists of a jammer detector that uses the estimate obtained at the first stage.
Finally, it is important to observe that MOS rules can be adapted to
accomplish detection tasks by also considering the model order ``$0$'' which is associated with the null hypothesis
\cite{kailathMultipleHypotheses,FishlerMultHyp,VanTrees4}. However, in this case, it is not possible to set
any threshold in order to guarantee a preassigned Probability of False Alarm ($P_{fa}$) and, more importantly,
the CFAR property, which is of primary concern in radar, cannot be a priori stated.

In this paper, we develop an elegant framework relying on the Kullback-Leibler Information Criterion (KLIC) 
\cite{kullback1951} to address multiple hypothesis testing problems where there exist many alternative hypotheses. This 
framework provides an important interpretation of both the Likelihood Ratio Test (LRT) and the GLRT from an information 
theoretic standpoint and, remarkably, it lays the theoretical foundation for the design of new one-stage decision 
schemes for multiple hypothesis tests. This result represents the main technical contribution of this paper and, 
interestingly, such new detection architectures share a common structure that is given by the sum of a conventional
decision statistic and a penalty term as well as the generic structure of a KLIC-based MOS rule consists of the 
compressed log-likelihood plus a penalty term.
The starting point of the developed framework is the case where all the parameters are known showing that 
under suitable regularity conditions the LRT approximates a test which selects the hypothesis with the associated probability density function (pdf) minimizing the Kullback-Leibler Divergence (KLD) with respect to the true data distribution. In addition, the LRT coincides with such test when the KLD is measured with respect to the Empirical Data Distribution (EDD).
Then, we guide the reader towards more difficult scenarios where the distribution parameters are no longer known. Specifically, we resort to the Taylor series approximations of the KLD, that are also used to derive the MOS rules, in order to come up with decision schemes capable of moderating the overfitting inclination of the MLA. From a different perspective, the same results can be obtained by regularizing the pdf under the generic alternative hypothesis through a suitable prior for the unknown model order and applying a procedure that combines the MLA and the Bayesian estimation \cite{jointMAP_ML,Stoica1} (see the appendix for further details).
The proposed theoretical framework is, then, completed with the investigation of the CFAR behavior of these decision schemes framing the analysis in the more general context of statistical invariance and providing two propositions 
which allows to state when the newly proposed decision architectures are invariant with respect to a given group 
of transformations and, possibly, enjoy the CFAR property. Finally, we present numerical examples obtained over simulated data and concerning three different radar detection problems also in comparison with two-stage architectures where the first stage is aimed at estimating the model order while the second stage is a conventional detector. The analyses highlight that the developed MOS-based detectors are capable of providing good detection capabilities in several contexts of practical interest.

The remainder of this paper is organized as follows: the next section contains preliminary definitions
and formulates the detection problem at hand in terms of a multiple hypothesis test. Section \ref{Sec:Architectures}
describes the proposed framework and introduces the new architectures, whereas in Section \ref{sec_invariance} some discussions about the CFAR properties ensured exploiting the Principle of Invariance are also provided. Then, Section \ref{Sec:Applications} provides some illustrative examples to assess the performance of the new architectures also in comparison to natural competitors. Concluding remarks and future research tracks are given in Section \ref{Sec:Conclusions}. The derivation of an alternative framework leading to the same architectures 
is confined to the Appendix.

\subsection{Notation}
In the sequel, vectors and matrices are denoted by boldface lower-case and upper-case letters, respectively.
The symbols $\det(\cdot)$, $\tr(\cdot)$, $(\cdot)^T$, $(\cdot)^\dag$, and $(\cdot)^{-1}$ denote the determinant, 
trace, transpose, conjugate transpose, and inverse respectively. As to numerical sets, $\R$ is the set of 
real numbers, $\R^{N\times M}$ 
is the Euclidean space of $(N\times M)$-dimensional real matrices (or vectors if $M=1$), $\C$ is the set of 
complex numbers, and $\C^{N\times M}$ is the Euclidean space of $(N\times M)$-dimensional complex matrices 
(or vectors if $M=1$). The cardinality of a set $\Omega$ is denoted by $|\Omega|$.
The Dirac delta function is indicated by $\delta(\cdot)$.
$\bI_N$ stands for the $N \times N$ identity matrix, while $\bzero$ is the null vector or matrix of proper size. The acronym i.i.d. mean probability density function and independent and identically  distributed, respectively, while symbol $E_f[\cdot]$ denotes the statistical expectation with respect to the pdf $f$.
If $A$ and $B$ are two continuous random variables, $f(A|B)$ is the conditional pdf of $A$ given $B$, whereas
the conditional probability of an event $A$ given the event $B$ is represented as $P(A|B)$.
Finally, we write $\bx\sim\cC\cN_N(\bm, \bM)$ if $\bx$ is a complex circular $N$-dimensional normal vector with mean $\bm$ and covariance matrix $\bM \succ \bzero$, whereas $\bx\sim f(\bx;\btheta)$ means that $f$ is the pdf of $\bx$ with parameter vector $\btheta$.

\section{Problem Formulation and Preliminary Definitions}
\label{Sec:Problem_Formulation}
Let us consider a radar system, equipped with $N$ space and/or time identical channels, which collects a 
data matrix $\bZ=[\bz_1,\ldots,\bz_K]\in\C^{N\times K}$ whose\footnote{Observe that in the case of Space-Time
Adaptive Processing (STAP), $N$ represents the number of space-time channels, whereas when the system transmits either a single pulse through an array of sensors or exploits a single antenna to transmit a pulse train, $N$ represents the 
number of elements of either the spatial array or the number of transmitted pulses, respectively \cite{Richards}. 
Finally, $K$ may represent the number of range bins or the number of pulses when a slice of the STAP datacube
is processed.} columns can be modeled as statistically independent random vectors whose distribution belongs to a 
preassigned family. For simplicity and in order not to burden the notation (a point better explained below), 
we assume that these vectors share the same distribution parameters (identically distributed) and that their 
joint unknown pdf is denoted by $\bar{f}(\bZ;\btheta)$ with $\btheta\in\R^{p\times 1}$ the parameter vector 
taking on value in a specific parameter space, $\Theta\subseteq\R^{p\times 1}$ say. Now, a conventional binary 
decision problem partitions the latter into two subsets $\Theta_0$ and $\Theta_1$ corresponding to 
the null ($H_0$) and the alternative ($H_1$) hypothesis, respectively. Thus, denoting by $\btheta_0$ 
the elements of $\Theta_0$ and by $\btheta_1$ the elements of $\Theta_1$, then the pdf of $\bZ$ under 
the $i$th hypothesis can be written as $f_i(\bZ;\btheta_i)$, $i=0,1$. It is important to stress here 
that $\bar{f}(\bZ;\btheta)$ is the actual pdf and is unknown, whereas $f_i(\bZ;\btheta_i)$ is the pdf 
of $\bZ$ when\footnote{In what follows, we assume correctly specified models, namely that data 
distribution family is known.} $\btheta=\btheta_i\in\Theta_i$.

Generally speaking, in many detection applications, the actual value of the parameter vector is unknown 
or, at least, partially known. In addition, not all the entries of $\btheta$ are useful to take a decision
for the specific problem at hand. As a consequence, we can split $\btheta$ as $\btheta=[\btheta_r^T, \btheta_s^T]^T$, 
where $\btheta_r\in\Theta_r\subseteq\R^{p_r\times 1}$ contains the parameters of interest, while 
the components of $\btheta_s\in\Theta_s\subseteq\R^{p_s\times 1}$ represent
the parameters that do not enter into the decision process and are called {\em nuisance parameters} \cite{KayBook}. 
From a more general perspective, it would be possible to associate a parameter vector of interest to 
each column of $\bZ$ with the consequence that $\bz_k$s are no longer identically distributed. However,
as stated before, we prefer to proceed assuming that $\bz_k$s share the same parameter vector in order to
maintain an easy notation and because the extension to the more general case is straightforward as we will
show in Section \ref{Sec:Applications}. Thus, with the above remarks in mind, a conventional binary hypothesis 
testing problem can be expressed as\footnote{Note that $p_r$ is the
number of the parameters of interest and $p_s$ is the number of the nuisance parameters. It follows that
$p=p_r+p_s$.}
\be
\left\{
\begin{array}{l}
\begin{split}
H_0: \ds\bZ \sim f_{0}(\bZ; \btheta_0) &= f_{0}(\bZ; \btheta_{r,0}, \btheta_{s})\\
&= \prod_{k=1}^K g_{0}(\bz_k; \btheta_{r,0}, \btheta_{s}),
\\
H_{1}: \ds\bZ \sim f_{1}(\bZ; \btheta_1) &= f_{1}(\bZ; \btheta_{r,1}, \btheta_{s})\\
&= \prod_{k=1}^K g_{1}(\bz_k; \btheta_{r,1}, \btheta_{s}),
\end{split}
\end{array}
\right.
\label{eqn:detectionProblem}
\ee
where $\btheta_i=[\btheta_{r,i}^T, \btheta_{s}^T]^T\in\R^{p\times 1}$, $i=0,1$, with 
$\btheta_{r,i}\in\Theta_r^{i}\subseteq\R^{p_{r}\times 1}$ the parameter vector of interest under $H_i$,
$g_{i}(\cdot;\cdot)$, $i=0,1$, is the pdf of $\bz_k$, $k=1,\ldots,K$, under $H_i$.

Two remarks are now in order. First, note that $\{\Theta_r^0,\Theta_r^1\}$ is a partition of $\Theta_r$. 
Second, the above problem assumes that $p_r$, $p_s$, and, hence, the total number of unknown parameters $p$, 
are perfectly known. However, in many radar (and, generally speaking, signal processing) applications, 
it is not seldom for $p$ to be unknown under the alternative hypothesis due to the fact that the size 
of $\btheta_{r,1}$ might depend on the specific operating scenario \cite{Stoica1}. 
For instance, radar systems might face with situations where an unknown number of targets are present 
in the surveillance area \cite{1605248,BR_multipleTargets} or be under the attack of an unknown number of 
jammers \cite{carotenutoNLJ,8781902,8835670}. Under this assumption, the hypothesis test can be modified as
\be
\begin{cases}
\mbox{under $H_0$}: \ \btheta_{r,0}\in\Theta_r^{0}\subseteq\R^{p_{r,0}\times 1},
\\
\mbox{under $H_1$}: \ \btheta_{r,1}\in\Theta_r^{m}\subseteq\R^{p_{r,m}\times 1}, \ p_{r,m}\in\Omega_r,
\end{cases}
\ee
where $p_{r,0}$ and $\Omega_r=\{ p_{r,1},\ldots, p_{r,M} \}$ with $p_{r,1}\leq\ldots\leq p_{r,M}$ are known, 
while $p_{r,m}$ is unknown. It follows that the pdf of $\bZ$ (and, hence, that of $\bz_k$) under
$H_1$ depends on $p_{r,m}$ and, more important, the uncertainty on the latter leads to a testing problem 
formed by multiple (possibly nested) $H_1$ hypotheses, i.e.,
\be
\begin{cases}
H_0 : & \ds\bZ \sim f_{0}(\bZ; \btheta_{r,0}, \btheta_{s}),
\\
H_{{1,1}} : & \bZ \sim f_{1,1}(\bZ; \btheta_{r,1}, \btheta_{s}, p_{r,1}),
\\
\vdots & \vdots
\\
H_{{1,M}} : & \bZ \sim f_{1,M}(\bZ; \btheta_{r,1}, \btheta_{s}, p_{r,M}),
\end{cases}
\label{eqn:MHP}
\ee
where $f_{1,m}(\bZ; \btheta_{r,1}, \btheta_{s}, p_{r,m})=\prod_{k=1}^K g_{1,m}(\bz_k;\btheta_{r,1}, 
\btheta_{s}, p_{r,m})$ is the pdf of $\bZ$ under $H_{1,m}$, namely when $\btheta_{r,1}\in\Theta_r^{m}$, 
with $g_{1,m}(\bz_k;\btheta_{r,1}, \btheta_{s}, p_{r,m})$ the pdf of $\bz_k$ under $H_{1,m}$.

In the next section, we propose an {\em Information-theoretic} based approach to deal with problem (\ref{eqn:MHP}) 
exploiting the KLIC \cite{kullback1951}. Specifically, this criterion
relies on the measurement of a certain distance, the so-called KLD, between a candidate distribution belonging to the family of densities
$\cF=\left\{ f_0, \ f_{1,m}, \ m\in\{ 1,\ldots, M\} \right\}$
and the actual distribution of $\bZ$, which is assumed to lie in $\cF$ and denoted by
\be
\bar{f}(\bZ;\btheta) = \prod_{k=1}^K \bar{g}(\bz_k;\btheta),
\ee
where $\bar{g}(\bz_k;\btheta)$ is the true pdf of $\bz_k$, $k=1,\ldots,K$.
Besides, we suppose that the inequalities required to invoke the Kintchine's Strong Law of Large Numbers \cite{senSinger} are valid, namely
\begin{align}
|E_{\bar{g}}[\log g_{1,m}(\bz_k; \btheta_{r,1}, \btheta_{s}, p_{r,m})]| &<+\infty,
\label{eqn:finteMeanH1}
\\
|E_{\bar{g}}[\log g_{0}(\bz_k; \btheta_{r,0}, \btheta_{s})]| &<+\infty.
\label{eqn:finteMeanH0}
\end{align}
as well as the following ``regularity conditions'' \cite{Stoica1}
\begin{equation}
\begin{split}
\frac{1}{T} &\frac{\partial^2}{\partial \btheta\partial\btheta^T}
\left[\log \prod_{k=1}^K f_{1,m}(\bz_k; \btheta_{r,1}, \btheta_{s}, p_{r,m})\right]\\
&\overset{T\rightarrow \infty}{\longrightarrow}
\frac{1}{T} E\left\{ \frac{\partial^2}{\partial \btheta\partial\btheta^T}
\left[\log \prod_{k=1}^K f_{1,m}(\bz_k; \btheta_{r,1}, \btheta_{s}, p_{r,m})\right]
\right\},
\end{split}
\end{equation}
and ${p}/{T}\overset{T\rightarrow \infty}{\longrightarrow} 0$,
which are required to suitably approximate the KLD. In the above equations, $T$ represents the total number of
real-valued observations, that, for the problem at hand, is equal to $2NK$ since we are dealing with complex vectors.

Note that a ``minimum information distance'' selection criterion has already been successfully applied 
for model order estimation giving rise to the so-called MOS rules \cite{Stoica1,selen2007model}. 
The resulting selection architectures share the same structure consisting of a fitting term 
(the compressed log-likelihood function) plus an adjustment which also depends on the number of parameters. 
As a consequence, the parameter number diversity comes into play to moderate the overfitting inclination 
of the compressed likelihood in the case of nested hypothesis. In fact, in this context, the KLIC-based 
rules can provide satisfactory classification performance whereas the Maximum Likelihood (ML) approach may
fail because the likelihood function monotonically increases with $p_{r,m}$ and the ML estimate (MLE) of $p_r$ 
will always be $p_{r,M}$ (or, equivalently, the MLE of $p$ will be $p_{r,M}+p_s$).

To conclude this preliminary section, for the reader ease, we recall that the KLD \cite{kullback1951} (also called
relative entropy) between $\bar{f}$ and $f_{1,m}$, namely the pdf of a generic candidate model under $H_1$, can be
written as \cite{cover2012elements}
\be
D(\bar{f},f_{1,m}) = 
\int_{-\infty}^{\infty} \bar{f}(\bZ) \log\frac{\bar{f}(\bZ)}{f_{1,m}(\bZ; \btheta_{r,1}, \btheta_{s}, p_{r,m})}d\bZ,
\label{eqn:definitionKLD}
\ee
where $d\bZ=dz^r_{1,1}dz^i_{1,1}\ldots dz^r_{N,K}dz^i_{N,K}$ with $z^r_{n,k}$ and $z^i_{n,k}$ the real
and imaginary parts of the $n$th component of $\bz_k$, and can be decomposed into the sum of two 
terms\footnote{We assume that the considered pdfs exist with respect to a Lebesgue measure.}
\begin{align}
D(\bar{f},f_{1,m}) &= \int_{-\infty}^{\infty} \bar{f}(\bZ) \log{\bar{f}(\bZ)}d\bZ \nonumber\\
&-\int_{-\infty}^{\infty} \bar{f}(\bZ) \log{f_{1,m}(\bZ; \btheta_{r,1}, \btheta_{s}, p_{r,m})}d\bZ \nonumber
\\
&=-h(\bar{f}) + h(\bar{f},f_{1,m}),
\label{eqn:KLdiscrepancySum_r}
\end{align}
where $h(\bar{f})$ is the differential entropy of $\bar{f}$ and $h(\bar{f},f_{1,m})$ is the cross entropy between 
$\bar{f}$ and $f_{1,m}$. Note that unlike $h(\bar{f})$, $h(\bar{f},f_{1,m})$ depends on the $m$th model 
(or hypothesis). Analogously, we can write the KLD with respect to the pdf under $H_0$ as
\be
D(\bar{f},f_{0}) = -h(\bar{f}) + h(\bar{f},f_{0}).
\label{eqn:KLdiscrepancySum_0}
\ee
Recall that $D(\cdot,\cdot)$ is not a true distance between distributions since it is not symmetric and does 
not satisfy the triangle inequality \cite{cover2012elements}. Nonetheless, it is often useful to think of the 
KLD as a ``distance'' between distributions. Finally, the KLD can be interpreted as the information loss 
when either $f_{1,m}$ or $f_0$ is used to approximate $\bar{f}$ \cite{Anderson}.

\section{KLIC-based Decision Rules}
\label{Sec:Architectures}
In this section, we exploit the KLIC to devise decision schemes for problem \eqref{eqn:MHP}. To this end, 
we proceed by considering the case where all the parameters are known and then evolve to more difficult cases
where the parameters become unknown. It is important to observe that we define here an information theoretic 
framework where well-established decision rules as the Likelihood Ratio Test and the GLRT are suitably encompassed.

\subsection{KLIC-based Detectors for Known Model and Parameters}
In this case, the number of alternative hypotheses is reduced to $1$ and problem \eqref{eqn:MHP} turns into 
problem \eqref{eqn:detectionProblem} with the additional assumption that $\btheta_0$ and $\btheta_1$ 
are known. Moreover, the number of parameters of interest is $p_{r,\bar{m}}$ and is known. As a consequence, 
data distribution is completely determined by the hypotheses and the true pdf of $\bZ$ belongs to the following 
family $\cF_{\theta,p}=\{ f_0, \ f_{1,\bar{m}} \}$.
Thus, a natural test based upon KLIC would decide for the hypothesis where the associated pdf (computed at $\bZ$) exhibits the ``minimum distance'' from $\bar{f}$. Otherwise stated, such test can be formulated as
\be
D(\bar{f},f_{0}) \test D(\bar{f},f_{1,\bar{m}}).
\label{eqn:decision01}
\ee
The above rule selects $H_0$ if the distance between $\bar{f}$ and $f_0$ is lower than that between $\bar{f}$ 
and $f_{1,\bar{m}}$. In the opposite case, it decides for $H_1$. Now, let us exploit \eqref{eqn:KLdiscrepancySum_r} 
and \eqref{eqn:KLdiscrepancySum_0} to recast \eqref{eqn:decision01} as
\be
-h(\bar{f}) + h(\bar{f},f_{0}) \test -h(\bar{f}) + h(\bar{f},f_{1,\bar{m}}),
\ee
or, equivalently,
\be
h(\bar{f},f_{0}) - h(\bar{f},f_{1,\bar{m}}) \test 0.
\label{eqn:decision01new}
\ee
From the information theory point of view, the above forms highlight that the considered decision rule minimizes the loss of information which occurs when $\bar{f}$ is approximated with $f_{1,\bar{m}}$ or $f_0$ \cite{Anderson}. Moreover, using \eqref{eqn:KLdiscrepancySum_r}, we can rewrite \eqref{eqn:decision01new} as
\begin{align}\label{eqn:LRT00}
& E_{\bar{f}}[\log f_{1,\bar{m}}(\bZ; \btheta_{r,1}, \btheta_{s}, p_{r,\bar{m}})] \nonumber\\
&- E_{\bar{f}}[\log f_{0}(\bZ; \btheta_{r,0}, \btheta_{s})] \test 0
\\
\Rightarrow &
\sum_{k=1}^K E_{\bar{g}}[\log g_{1,\bar{m}}(\bz_k; \btheta_{r,1}, \btheta_{s}, p_{r,\bar{m}})]\nonumber\\
& - \sum_{k=1}^K 
E_{\bar{g}}[\log g_{0}(\bz_k; \btheta_{r,0}, \btheta_{s})] \test 0
\\
\Rightarrow &
K E_{\bar{g}}[\log g_{1,\bar{m}}(\bz; \btheta_{r,1}, \btheta_{s}, p_{r,\bar{m}})]\nonumber \\
& - K E_{\bar{g}}[\log g_{0}(\bz; \btheta_{r,0}, \btheta_{s})] \test 0
\label{eqn:LRT_third}
\\
\Rightarrow
& E_{\bar{g}}[\log g_{1,\bar{m}}(\bz; \btheta_{r,1}, \btheta_{s}, p_{r,\bar{m}})] \nonumber\\
&- E_{\bar{g}}[\log g_{0}(\bz; \btheta_{r,0}, \btheta_{s})] \test 0
\label{eqn:KLDtest_zk}
\\
\Rightarrow & \Lambda(\bZ;\btheta_{r,1}, \btheta_{r,0}, \btheta_{s}, \bar{m}) = 
\nonumber\\
& E_{\bar{g}}\left[\log \frac{g_{1,\bar{m}}(\bz; \btheta_{r,1}, \btheta_{s}, p_{r,\bar{m}})}
{g_{0}(\bz; \btheta_{r,0}, \btheta_{s})}\right] \test 0.
\label{eqn:LRT01}
\end{align}
Note that starting from \eqref{eqn:LRT_third}, we have dropped the subscript of $\bz_k$ since $\bz_1,\ldots,\bz_K$ 
are i.i.d. random vectors. Test \eqref{eqn:LRT01} cannot be applied in practice because it involves the 
computation of the expected log-likelihood ratio and requires the knowledge of $\bar{g}(\cdot)$ (or, equivalently, 
of $\bar{f}(\cdot)$). For this reason, we replace $\Lambda$ with a suitable estimate which is function of 
the observed data. To this end, we resort to the following unbiased estimator \cite{EGUCHI20062034}
\be
\begin{split}
\widehat{\Lambda}(\bZ;\btheta_{r,1}, \btheta_{r,0}, \btheta_{s}, \bar{m}) &= \frac{1}{K} \sum_{k=1}^K 
\log \frac{g_{1,\bar{m}}(\bz_k; \btheta_{r,1}, \btheta_{s}, p_{r,\bar{m}})}{g_{0}(\bz_k; \btheta_{r,0}, \btheta_{s})} \\
& = \frac{1}{K} \log \frac{f_{1,\bar{m}}(\bZ; \btheta_{r,1}, \btheta_{s}, p_{r,\bar{m}})}
{f_{0}(\bZ; \btheta_{r,0}, \btheta_{s})}.
\end{split}
\ee
In fact, since the technical assumptions of the Kintchine's Strong Law of Large Numbers hold true due to
\eqref{eqn:finteMeanH1} and \eqref{eqn:finteMeanH0}, in the limit for $K\rightarrow +\infty$, we have that
$\widehat{\Lambda}$ converges almost surely to $\Lambda$, namely
\be
\widehat{\Lambda}(\bZ;\btheta_{r,1}, \btheta_{r,0}, \btheta_{s}, \bar{m}) \xrightarrow{a.s.} 
\Lambda(\bZ;\btheta_{r,1}, \btheta_{r,0}, \btheta_{s}, \bar{m}).
\ee
Summarizing, test \eqref{eqn:LRT01} is replaced by
\be
\widehat{\Lambda}(\bZ;\btheta_{r,1}, \btheta_{r,0}, \btheta_{s}, \bar{m}) \test 0,
\label{eqn:LRT02}
\ee
where $\widehat{\Lambda}$ is a random variable whose value depends on the observed data.
An alternative derivation for \eqref{eqn:LRT02} consists in replacing $f_{1,\bar{m}}(\cdot)$ and $f_{0}(\cdot)$ with 
$g_{1,\bar{m}}(\cdot)$ and $g_0(\cdot)$, respectively, in \eqref{eqn:decision01}, while the EDD, 
whose expression is $\varXi(\bz)=\frac{1}{K}\sum_{k=1}^K \delta(\bz-\bz_k)$,
is used in place of $\bar{f}(\cdot)$. Thus, the KLD between the EDD and the distribution of a generic $\bz_k$ 
is measured to decide which hypothesis is in force. As a matter of fact, criterion (\ref{eqn:decision01}) becomes
\begin{align}
& D(\varXi,g_0) \test D(\varXi,g_{1,\bar{m}})
\label{eqn:decisionEDD_00}
\\
\Rightarrow &E_{\varXi}[\log g_{1,\bar{m}}(\bz; \btheta_{r,1}, \btheta_{s}, p_{r,\bar{m}})] \nonumber \\
&- E_{\varXi}[\log g_{0}(\bz; \btheta_{r,0}, \btheta_{s})] \test 0
\\
\Rightarrow &\sum_{k=1}^K\log g_{1,\bar{m}}(\bz_k; \btheta_{r,1}, \btheta_{s}, p_{r,\bar{m}}) \nonumber \\
&- \sum_{k=1}^K\log g_{0}(\bz_k; \btheta_{r,0}, \btheta_{s}) \test 0
\\
\Rightarrow &\log \frac{\ds\prod_{k=1}^K g_{1,\bar{m}}(\bz_k; \btheta_{r,1}, \btheta_{s}, p_{r,\bar{m}})} 
{\ds\prod_{k=1}^K g_{0}(\bz_k; \btheta_{r,0}, \btheta_{s})} \test 0.
\end{align}
Finally, note that the detection threshold of test \eqref{eqn:LRT02} is set to zero and, as a consequence, it 
does not allow for a control on the probability of type I error also known as 
$P_{fa}$ in Detection Theory \cite{KayBook}. 
In order to circumvent this limitation, the decision rule can be modified as
\be
\log \frac{f_{1,\bar{m}}(\bZ; \btheta_{r,1}, \btheta_{s}, p_{r,\bar{m}})}
{f_{0}(\bZ; \btheta_{r,0}, \btheta_{s})} \test \eta,
\label{eqn:NP-test}
\ee
where the detection threshold, $\eta$, is suitably tuned in order to guarantee the 
desired\footnote{Hereafter, symbol $\eta$ is used to denote the generic detection threshold.} $P_{fa}$. 
Remarkably, the decision rule \eqref{eqn:NP-test} is statistically equivalent to the 
LRT or the Neyman-Pearson test \cite{lehmann1986testing}.

\subsection{KLIC-based Detectors for Known Model and Unknown Parameters}
Let us now consider \eqref{eqn:LRT00} and assume that only the number of parameters of interest, $p_{r,\bar{m}}$ say, 
is known, whereas $\btheta_{r,1}$, $\btheta_{r,0}$, and $\btheta_{s}$ are unknown. In this case the family of 
candidate models becomes
\be
\begin{split}
\cF_{p} & = \cF_{0} \cup \cF_{1} = \{f_0(\cdot;\btheta_{r,0},\btheta_s): \btheta_{r,0}\in \Theta_r^0, \ 
\btheta_{s}\in \Theta_s \}\\
& \cup \{f_{1,\bar{m}}(\cdot;\btheta_{r,1},\btheta_s,p_{r,\bar{m}}): \ \btheta_{r,1}\in \Theta_r^{\bar{m}}, \ 
\btheta_{s}\in \Theta_s \},
\end{split}
\ee
where $\Theta_r^0$ and $\Theta_r^{\bar{m}}$ form a partition of the parameter space of interest while $\Theta_s$ 
is the nuisance parameter space. Note that in this case, the hypotheses of \eqref{eqn:detectionProblem} are 
composite implying that, in order to build up a decision rule based upon \eqref{eqn:decision01}, the unknown 
parameters $\btheta_{r,1}$, $\btheta_{r,0}$, and $\btheta_{s}$ must be estimated from data. Among different
alternatives, we resort to the ML approach, which enjoys ``good'' asymptotic properties \cite{KayBook_Estimation}. 
In fact, given a model, the consistency (when it holds) of the MLE 
ensures that it converges in probability to the true parameter value, which is also the minimizer of the 
KLD (see \cite{Anderson,geyer} and references therein).
Thus, in \eqref{eqn:LRT00}, we replace $\btheta_{r,1}$, $\btheta_{r,0}$, and $\btheta_{s}$ with their 
respective MLEs under each hypothesis. Specifically, denoting by $\widehat{\btheta}_{r,i}$, $i=0,1$, the MLE of
$\btheta_{r,i}$ and by $\widehat{\btheta}_{s,i}$ the MLE of $\btheta_{s}$ under the $H_i$ hypothesis, $i=0,1$,
\eqref{eqn:LRT00} can be recast as
\begin{align}
& E_{\bar{f}}[\log f_{1,\bar{m}}(\bZ; \widehat{\btheta}_{r,1}, \widehat{\btheta}_{s,1}, p_{r,\bar{m}})]\nonumber
\\
& - E_{\bar{f}}[\log f_{0}(\bZ; \widehat{\btheta}_{r,0}, \widehat{\btheta}_{s,0})] \test 0 \nonumber
\\
& \Rightarrow \Lambda_{1}(\bZ;\bar{m}) = E_{\bar{f}}\left[\log \frac{f_{1,\bar{m}}(\bZ; \widehat{\btheta}_{r,1}, 
\widehat{\btheta}_{s,1}, p_{r,\bar{m}})}{f_{0}(\bZ; \widehat{\btheta}_{r,0}, \widehat{\btheta}_{s,0})}\right] \test 0.
\label{eqn:GLRT00}
\end{align}
Now, in place of the expectation with respect to the unknown $\bar{f}$, we use an unbiased estimator 
of $\Lambda_{1}(\bZ;\bar{m})$ (see also equation (48) of \cite{Stoica1}), namely
\be
\label{eqn:GLRT000}
\widehat{\Lambda}_{1}(\bZ;\bar{m}) = \log\frac{f_{1,\bar{m}}(\bZ; \widehat{\btheta}_{r,1}, 
\widehat{\btheta}_{s,1}, p_{r,\bar{m}})} {f_{0}(\bZ; \widehat{\btheta}_{r,0}, \widehat{\btheta}_{s,0})},
\ee
and introduce a threshold to control the $P_{fa}$ yielding
\be
\widehat{\Lambda}_{1}(\bZ;\bar{m})\test\eta.
\label{eqn:GLRT01}
\ee
It is important to underline that the above test is statistically equivalent to GLRT for known $p_{r}$.

The same result can be derived showing that the MLEs minimize the KLD between the EDD and the candidate 
model with respect to the unknown parameters. To this end, let us start from \eqref{eqn:decisionEDD_00}
and minimize both sides with respect to the unknown parameters
\be
\min_{\btheta_{r,0},\btheta_{s}} D(\varXi,g_{0}) \test \min_{\btheta_{r,1},\btheta_{s}} D(\varXi,g_{1,\bar{m}}).
\label{eqn:minKLD_EDD}
\ee
The above problem is equivalent to
\begin{align}
&  \min_{\btheta_{r,0}, \btheta_{s}} \{-E_{\varXi}[\log g_{0}(\bz;\btheta_{r,0}, \btheta_{s})]\} \test \nonumber
\\ &\min_{\btheta_{r,1}, \btheta_{s}} \{-E_{\varXi}[\log g_{1,\bar{m}}(\bz; \btheta_{r,1}, \btheta_{s}, p_{r,\bar{m}})]\}
\end{align}
\begin{align}
& \Rightarrow \underbrace{ \min_{\btheta_{r,0}, \btheta_{s}} \left[-\sum_{k=1}^K \log g_{0}(\bz_k;\btheta_{r,0}, \btheta_{s})\right]}_{P_1} \test \nonumber
\\ &\underbrace{\min_{\btheta_{r,1}, \btheta_{s}} \left[-\sum_{k=1}^K 
\log g_{1,\bar{m}}(\bz_k;\btheta_{r,1}, \btheta_{s}, p_{r,\bar{m}})\right]}_{P_2}.
\label{eqn:minKLD_EDD_01}
\end{align}
Now, note that $P_1$ and $P_2$ are equivalent to
\begin{align}
&-\max_{\btheta_{r,0}, \btheta_{s}} \sum_{k=1}^K \log g_{0}(\bz_k;\btheta_{r,0}, \btheta_{s}),
\\
&-\max_{\btheta_{r,1}, \btheta_{s}} \sum_{k=1}^K \log g_{1,\bar{m}}(\bz_k;\btheta_{r,1}, \btheta_{s},p_{r,\bar{m}}),
\end{align}
respectively, and, hence, the sought minimizers coincides with the MLEs of $\btheta_{r,1}$, $\btheta_{r,0}$,
$\btheta_{s}$ (the latter under each hypothesis), namely
$\widehat{\btheta}_{r,1} = \arg\min_{\btheta_{r,1}} D(\varXi,g_{1,\bar{m}})$, 
$\widehat{\btheta}_{s,1} = \arg\min_{\btheta_{s}} D(\varXi,g_{1,\bar{m}})$,
$\widehat{\btheta}_{r,0} = \arg\min_{\btheta_{r,0}} D(\varXi,g_{0})$, and
$\widehat{\btheta}_{s,0} = \arg\min_{\btheta_{s}} D(\varXi,g_{0})$.
With the above result in mind, \eqref{eqn:minKLD_EDD_01} can be recast as
\begin{align}
& \sum_{k=1}^K \log g_{1,\bar{m}}(\bz_k;\widehat{\btheta}_{r,1},\widehat{\btheta}_{s,1},p_{r,\bar{m}})\nonumber
\\ & -  
\sum_{k=1}^K \log g_{0}(\bz_k;\widehat{\btheta}_{r,0},\widehat{\btheta}_{s,0}) \test 0 \nonumber
\\
\Rightarrow &\log \frac{\ds \prod_{k=1}^K 
g_{1,\bar{m}}(\bz_k; \widehat{\btheta}_{r,1},\widehat{\btheta}_{s,1}, p_{r,\bar{m}})} 
{\ds \prod_{k=1}^K g_{0}(\bz_k; \widehat{\btheta}_{r,0},\widehat{\btheta}_{s,0})} \test 0,
\end{align}
which coincides with \eqref{eqn:GLRT000}.

\subsection{KLIC-based Detectors for Unknown Model and Parameters}
In this last subsection, we deal with the most general case where $p_r$, $\btheta_{r,1}$, $\btheta_{r,0}$, 
and $\btheta_{s}$ are unknown. As stated in Section \ref{Sec:Problem_Formulation}, under this assumption, 
there exist multiple alternative hypotheses depending on the model order $p=p_{r,m}+p_s$.

As possible strategy to select the most likely hypothesis, we might follow the same line of reasoning 
as in the previous subsection replacing $p$ with its MLE. However, if on one hand, this approach (given $p_r$) 
makes sense for $\btheta_{r,1}$, $\btheta_{r,0}$, and $\btheta_{s}$, on the other, when the considered models 
are nested, it fails in the estimation of $p_r$. 
In fact, the log-likelihood function monotonically increases with $p_r$ and, as a consequence, the ML approach 
will always select the maximum possible order leading to an overfitting. Therefore, an alternative approach 
is required in order to find ``good'' approximations of the negative cross entropy which moderate the 
overfitting inclination of the ML approach. To this end, we follow the same line of reasoning used to derive 
the MOS rules as shown in \cite{Stoica1}, where suitable Taylor series expansions of the cross entropy 
(used in \eqref{eqn:LRT00}) are exploited. Then, the dependence on $p$ is removed by optimizing these 
expansions over the latter, which is tantamount to minimizing approximations of the KLD between $\bar{f}$ 
and $f_{1,m}$ with respect to the unknown parameters. 

As for the null hypothesis, it is independent of $p$ and, hence, we can exploit previously devised estimators.
Specifically, we replace $E_{\bar{f}}[\log f_{0}(\bZ; {\btheta}_{r,0}, {\btheta}_{s})]$
with the same estimator as in the previous subsection, namely $\log f_{0}(\bZ; \widehat{\btheta}_{r,0},
\widehat{\btheta}_{s,0})$.

Now, let us define $\cI_m(\btheta_{r,1}, \btheta_{s}) = -h(\bar{f},f_{1,m})$ and denote by $\widehat{\cI}_m$
an estimator of the former. Following the lead of \cite{Stoica1}, several
alternatives are possible for $\widehat{\cI}_m$, namely
\begin{itemize}
\item through equations (49)-(53) of \cite{Stoica1}, we come up with
\be
\widehat{\cI}_m=\log f_{1,m}(\bZ; \widehat{\btheta}_{r,1}, \widehat{\btheta}_{s,1}, p_{r,m}) 
- \frac{ p_{r,m}+p_s }{2};
\label{eqn:NN_based}
\ee
\item through equations (57)-(62) of \cite{Stoica1}, we obtain
\be
\widehat{\cI}_m = \log f_{1,m}(\bZ; \widehat{\btheta}_{r,1}, \widehat{\btheta}_{s,1}, p_{r,m}) - (p_{r,m}+p_s);
\label{eqn:AIC_based}
\ee
\item through equations (59)-(60) and (73) of \cite{Stoica1}, we have that
\be
\begin{split}
\widehat{\cI}_m = &\log f_{1,m}(\bZ; \widehat{\btheta}_{r,1}, \widehat{\btheta}_{s,1}, p_{r,m}) 
- \frac{1+\rho}{2}(p_{r,m}+p_s),\\
& \rho>1;
\end{split}
\label{eqn:GIC_based}
\ee
\item through equations (79)-(86) of \cite{Stoica1}, we get
\be
\widehat{\cI}_m = \log f_{1,m}(\bZ; \widehat{\btheta}_{r,1}, \widehat{\btheta}_{s,1}, p_{r,m}) 
- \frac{p_{r,m}+p_s}{2}\log(T) + C,
\label{eqn:BIC_based}
\ee
where $C$ is a constant. Note that the above expression results from an asymptotic approximation for 
sufficiently large values of $T$ of the more general form \cite{Stoica1}
\be
\widehat{\cI}_m  = \log f_{1,m}(\bZ; \widehat{\btheta}_{r,1}, \widehat{\btheta}_{s,1}, p_{r,m}) - \frac{1}{2} \log\det\widehat{\bJ} + C,
\label{eqn:BIC_full}
\ee
where
\be
\widehat{\bJ}=
\left. \left[ -\frac{\partial^2}
{\partial\btheta\partial\btheta^T} \log f_{1,m}(\bZ; {\btheta}_{r,1}, 
{\btheta}_{s,1}, p_{r,m}) \right]\right|_{{\btheta}_{r,1}=\widehat{\btheta}_{r,1} \atop 
{\btheta}_{s,1}=\widehat{\btheta}_{s,1}}
\ee
It is clear that other approximations are possible by considering the asymptotic behavior with respect to different
parameters.
\end{itemize}
Finally, an estimate of $m$ can be obtained as
\be
\widehat{m}=\argmax_{m\in\{1,\ldots,M\}} \widehat{\cI}_m
\ee
with $p_{r,\widehat{m}}$ the corresponding estimate of $p_r$, and we can replace each addendum 
of \eqref{eqn:LRT00} with the respective approximation to come up with the following adaptive rule
\be
\widehat{\cI}_{\widehat{m}}  - \log f(\bZ; \widehat{\btheta}_{r,0}, \widehat{\btheta}_{s,0}) \test 0,
\ee
which, introducing the threshold to control the $P_{fa}$, can be recast as
\be
\label{eqn:decisionCriterion}
\dmax_{m\in\{1,\ldots,M\}}\left\{\widehat{\Lambda}_{1}(\bZ;m) - h(m)\right\}\test\eta,
\ee
where
\be
h(m)=
\begin{cases}
(p_{r,m}+p_s)/2, & \mbox{(a)}
\\
(p_{r,m}+p_s), & \mbox{(b)}
\\
\ds \frac{1+\rho}{2}(p_{r,m}+p_s), \ \rho>1,  & \mbox{(c)}
\\
\ds\frac{(p_{r,m}+p_s)}{2}\log(T), & \mbox{(d)}
\end{cases}
\label{eqn:penaltyMOSDet}
\ee
is a penalty term.

Before concluding this section, an important remark is in order. Specifically, let us focus on
\eqref{eqn:decisionCriterion} and observe that $\widehat{\Lambda}_{1}(\bZ;m)$ is the logarithm of the generalized likelihood ratio (GLR) assuming that the model order is $p_{r,m}+p_s$. Thus, for sufficiently large $T$, it would be 
reasonable to consider alternative decision statistics which share the same asymptotic behavior as the GLR. 
For instance, the GLR can be replaced by the decision statistics of the Rao or Wald test \cite{KayBook}. Summarizing,
decision rule \eqref{eqn:decisionCriterion} represents a starting point for the design of detection
architectures dealing with multiple alternative hypotheses. Another modification of \eqref{eqn:decisionCriterion}
may consist in replacing $\widehat{\Lambda}_{1}(\bZ;m)$ with the decision statistics derived applying
{\em ad hoc} GLRT-based design procedures \cite{robey1992cfar,Yuri01,HaoSP_HE,BOR-Morgan,8781902}. 
Finally, in the Appendix, we show that \eqref{eqn:decisionCriterion} can be also
considered as the result of the joint application of the Bayesian and ML framework after assigning a suitable 
prior to the number of parameters.

\section{Invariance Issues and CFAR Property}
\label{sec_invariance}
The design of architectures capable of guaranteeing the CFAR property is an issue of primary concern in radar
(as well as in other application fields) since it allows for reliable target detection also in those
situations where the interference (or unwanted components) spectral properties are unknown or highly variable. 
As a matter of fact, controlling and keeping low the number of false alarms is a precaution
aimed at avoiding the disastrous effects that the latter may ensue. Thus, system engineers set
detection thresholds in order to guarantee very small values for the $P_{fa}$ \cite{farina1986review, rohling1983radar, barkat1989cfar, conte2002cfar, roman2000parametric, gini2002covariance}. 
Unfortunately, the CFAR property is not granted by a generic detection scheme and, hence, before
claiming the CFARness for a given receiver, it must be proved that its decision statistic does not depend on
the interference parameters under the null hypothesis. However, there exist some cases where the decision 
statistic is functionally invariant to those transformations that modify the nuisance parameters, 
which do not enter the decision process, and, at the same time, preserve the hypothesis testing 
problem. As a consequence, under the null hypothesis, such statistic can be expressed 
as a function of random quantities whose distribution is independent of the nuisance parameters ensuring 
the CFAR property \cite{kelly1986adaptive,robey1992cfar,BOR-Morgan}.

Therefore, the above remarks suggest that it may be reasonable to look for decision rules that are invariant 
to nuisance parameters in the sense described above. To this end, we can invoke the {\em Principle of Invariance} 
\cite{LehmannBook,scharf1991statistical}, whose key idea consists in finding a specific group of transformations, 
formed by a set $\cG$ equipped with a binary composition operation $\circ$, that leaves unaltered the formal 
structure of the hypothesis testing problem (also inducing a group of 
transformations in the parameter space) and the family of distribution under each hypothesis. Then, a 
data compression can be accomplished by resorting to the so-called {\em maximal invariant statistics} which organize 
data into distinguishable classes of equivalence with respect to the group of transformations wherein such 
statistics are constant. Now, given a group of transformations, every invariant test can be written in terms 
of the maximal invariant statistic whose distribution may depend only on a function of the parameters (induced
maximal invariant). If the latter exists and is constant over $\Theta_0$, then any invariant test guarantees 
the CFAR property with respect to the unknown nuisance parameters.

In what follows, we provide two propositions which allows to state when \eqref{eqn:decisionCriterion} 
is invariant with respect to a given group of transformations and, possibly, enjoys the CFAR property.

\medskip

\begin{prop}\label{PropCFAR}
Let us assume that there exists a group of transformations $\cL=(\cG,\circ)$, which acts on data through $l(\cdot)$
and leaves both the following binary hypothesis testing problems
\be
\cP_m: \left\{
\begin{array}{ll}
H_0 : & \ds\bZ \sim f_{0}(\bZ; \btheta_{r,0}, \btheta_{s}),
\\
H_{{1,m}} : & \bZ \sim f_{1,m}(\bZ; \btheta_{r,1}, \btheta_{s}, p_{r,m}),
\end{array}
\right.
\ee
for all $m\in\{1,\ldots,M\}$ and the data distribution family unaltered, then the decision statistic 
\eqref{eqn:decisionCriterion} is invariant with respect to $\cL$.
\end{prop}
\begin{proof}
Since each problem $\cP_m$ is invariant with respect to $\cL$ by definition, the GLR for the $m$th testing problem, 
namely $\widehat{\Lambda}_{1}(\bZ;m)$, is invariant to the same transformation group as shown in
\cite{LehmannBook,eaton1983multivariate}, namely
\be
\widehat{\Lambda}_{1}(l(\bZ);m) = \widehat{\Lambda}_{1}(\bZ;m),\quad \forall m\in\{1,\ldots,M\}.
\ee
As a consequence, we obtain that
\be
\begin{split}
& \dmax_{m\in\{1,\ldots,M\}}\left\{\widehat{\Lambda}_{1}(l(\bZ);m) - h(m)\right\}\\
& = \dmax_{m\in\{1,\ldots,M\}}\left\{\widehat{\Lambda}_{1}(\bZ;m) - h(m)\right\}.
\end{split}
\ee
The last equality establishes the invariance of \eqref{eqn:decisionCriterion} with respect to $\cL$ and 
concludes the proof.
\end{proof}

\medskip

From a practical point of view, if, for all $m\in\{1,\ldots,M\}$, the generic $\cP_m$ is invariant with  respect to a subgroup $\cL_m$ of a more general group, then it is possible to obtain $\cL$ as the intersection of the subgroups $\cL_1,\ldots,\cL_M$ \cite{robinson1996course}. 
Moreover, observe that if \eqref{eqn:decisionCriterion} is invariant with respect to $\cL$ then, by {\em Theorem 6.2.1} of \cite{LehmannBook}, its decision statistic can be expressed as a function of the previously described maximal invariant statistics. Next, under the hypothesis of {\em Theorem 6.3.2} of \cite{LehmannBook}, the distribution of a maximal invariant statistic depends on the induced maximal invariant which is denoted by $\xi(\btheta)$ and the following proposition holds true.

\medskip

\begin{prop}
Let us assume that {\em \bf Proposition \ref{PropCFAR}} is valid and that 
\be
\forall \btheta\in\Theta_0: \ \xi(\btheta)=c,
\ee
with $c\in\R$, then \eqref{eqn:decisionCriterion} ensures the CFAR property. 
\end{prop}
\begin{proof}
Since \eqref{eqn:decisionCriterion} is invariant, its decision statistic is a function of the maximal invariant
statistic, whose distribution does not depend on the specific value of the parameter vector under $H_0$ but only 
on $c$.
\end{proof}

\medskip

\section{Illustrative Examples}
\label{Sec:Applications}
In this section, we show the effectiveness of the newly proposed approach in three operating scenarios
related to radar systems. For each scenario, we compare the proposed solutions with Two-Stage (TS)
architectures consisting of an estimation stage devoted to the model order selection and a detection 
stage which exploits the model order estimate provided by the first stage\footnote{Some of these competitors
can be found in the open literature.}.

In what follows, the performance of the proposed approach is investigated in terms of the Probability of 
Correct Detection, $P_{d|m}$ say, defined as\footnote{Note that, focusing on problem (\ref{eqn:detectionProblem}) 
and assuming that $p_r$ 
is unknown, the $P_{d|m}$ can be recast as $P(H_1, \hat{m}=m | H_m)$ with $m$ the actual hypothesis order.}  
$P_{d|m} = P(H_m|H_m)$
as well as either
the classification histograms (a point better explained in the next
subsections) or the Root Mean Square Errors (RMSEs) of the interest parameters. 
Moreover, the above performance metrics are computed resorting to standard Monte Carlo counting 
techniques over $10^4$ and $100/P_{fa}$ independent trials to estimate the $P_{d|m}$ and the thresholds 
to guarantee a preassigned $P_{fa}=P(H_m, m>0|H_0)$, respectively. Finally,
all the numerical examples assume $K=32$, $P_{fa}=10^{-4}$, a thermal noise power $\sigma^2_n=1$, and
a radar system equipped with $N=16$ spatial channels
leading to the following expression for the nominal steering vector \cite{BOR-Morgan}
$\bv(\theta) = \frac{1}{4} \left[1, e^{j\pi \sin\theta}, \ldots, e^{j\pi 15 \sin\theta} \right]^T$,
where $\theta$ is the steering angle measured with respect to the antenna boresight.

\subsection{Multiple Noise-like Jammers Detection}
Let us assume that the considered radar system is listening to the environment in the presence of an unknown 
number of NLJs. 
Samples of interest are organized into $N$-dimensional vectors denoted by $\bz_k$, $k=1,\ldots,K$. 
Note that, in this case, data under test are not affected by clutter since they are 
collected without transmitting any signal \cite{DoppioTraining,carotenutoNLJ}. 
Thus, the detection problem at hand can be 
formulated in terms of the following multiple hypothesis test
\be
\label{probfornoisejam}
\left\{
\begin{array}{lll}
H_0: & \bz_k \sim \cC\cN_N(\bzero, \sigma^2_n\bI), & \! k = 1, \ldots, K, 
\\
H_{m}: & \bz_k \sim \cC\cN_N(\bzero, \sigma^2_n\bI + \bM_J(m)), & \! k = 1, \ldots, K, 
\\
& & \! m=1,\ldots,N_J,
\end{array}
\right.
\ee
where $\bM_J(m)\in\C^{N\times N}$ is the unknown positive semidefinite matrix representing the jammer component
whose rank is denoted by $m$ and is related to the actual 
number of directions from which interfering signals hit the system, $N_J\leq N$ is the maximum allowable number 
of such directions, and $\sigma^2_n \bI$ with unknown $\sigma^2_n>0$ is the noise component.

In order to apply \eqref{eqn:decisionCriterion}, we need to compute the logarithm of the 
GLR for \eqref{probfornoisejam} and the penalty term. 
The former is computed by following the lead of \cite{carotenutoNLJ},
where the role of $r$ is played by $m$, to come up with
\be
\begin{split}
&\widehat{\Lambda}_{1}(\bZ;m)
= \log \left\{ \frac{\displaystyle{\left[\frac{1}{K(N-m)} \sum_{i=m+1}^{N} \gamma_i \right]^{-K(N-m)}}}{\displaystyle{\prod_{i=1}^m \left(\frac{\gamma_i}{K}\right)^{K} \left[\frac{1}{NK} \sum_{i=1}^{N} \gamma_i \right]^{-NK}}} \right\},\\
&\quad m = 1,\ldots, N_J.
\end{split}
\ee
where $\bZ=[\bz_1,\ldots,\bz_K]$ and $\gamma_1\geq\ldots \geq\gamma_N\geq 0$ are the eigenvalues of $\bZ\bZ^\dag$.
The penalty term can be written by noticing that the number of unknown parameters is
$p=m(2N-m)+1$ while $T=2KN$. Notice that $T$ depends on $N$ and, hence, we exploits \eqref{eqn:BIC_full}
to find another suitable asymptotic approximation in place of \eqref{eqn:penaltyMOSDet}(d). Specifically, we assume that only $K$
goes to infinite obtaining $h(m)=\frac{p_{r,m}+p_s}{2}\log(K)$ (in the following we refer to this penalty term
as modified \eqref{eqn:penaltyMOSDet}(d)).

The considered simulation scenario comprises three noise jammers with different Angle Of Arrivals (AOA), 
viz., $10^{\circ}$, $20^{\circ}$, and $-15^{\circ}$, 
respectively, and sharing the same power. Therefore, the jammer component of the data covariance matrix 
can be written as $\bM_J = \text{JNR}\sum_{i=1}^3 \bv(\theta_i)\bv^{\dag}(\theta_i)$,
where $\theta_i$ is the AOA of the $i$th noise-like jammer and $\text{JNR}$ the Jammer to Noise Ratio.
Finally, we set $N_J=6$. As already stated, for 
comparison purposes, we show the performance of analogous TS architectures where the MOS rule at
the first stage shares up to a factor $2$ the same penalty term as \eqref{eqn:decisionCriterion} whereas 
the second stage is given by the GLRT for known $m$ \cite{carotenutoNLJ}.

Figure \ref{fignoise}, where we plot the $P_{d|3}$ versus JNR, shows that the curves related 
to the new decision schemes are perfectly overlapped to those 
of the TS architectures. The worst performance is provided by 
\eqref{eqn:decisionCriterion} coupled with \eqref{eqn:penaltyMOSDet}(a) due to the nonnegligible attitude 
of the latter to overestimate the hypothesis order as shown in the next figure. Such behavior, although
in a milder form, can also be noted for \eqref{eqn:decisionCriterion} coupled with \eqref{eqn:penaltyMOSDet}(b) whose 
detection curve exhibits a floor. On the other hand, \eqref{eqn:decisionCriterion} coupled with
\eqref{eqn:penaltyMOSDet}(c) and \eqref{eqn:decisionCriterion} coupled with modified \eqref{eqn:penaltyMOSDet}(d) provide satisfactory performance.

Finally, in Figure \ref{fignoisehist}, we plot the classification histograms, namely the probability to select 
$H_m$, $m=1,2,3,4,5,6$, under $H_n$, $n=1,2,3$ and assuming $\text{JNR}=10$ dB. It clearly turns out
the significant overestimation inclination of (\ref{eqn:decisionCriterion}) coupled with \eqref{eqn:penaltyMOSDet}(a).
The remaining decision schemes exhibit a percentage of correct classification 
very close to $100$\% except for \eqref{eqn:decisionCriterion} coupled with \eqref{eqn:penaltyMOSDet}(b)
whose percentages are around $90$ \%.

Other results not reported here for brevity highlight that \eqref{eqn:decisionCriterion} coupled with 
\eqref{eqn:penaltyMOSDet}(c) and \eqref{eqn:decisionCriterion} coupled with modified \eqref{eqn:penaltyMOSDet}(d) maintain good 
detection and classification 
performance for different parameter settings whereas the behavior of the other competitors is more sensitive 
to the specific parameter values experiencing in some situations a significant performance degradation.

\begin{figure}
\centering
\includegraphics[width=0.45\textwidth]{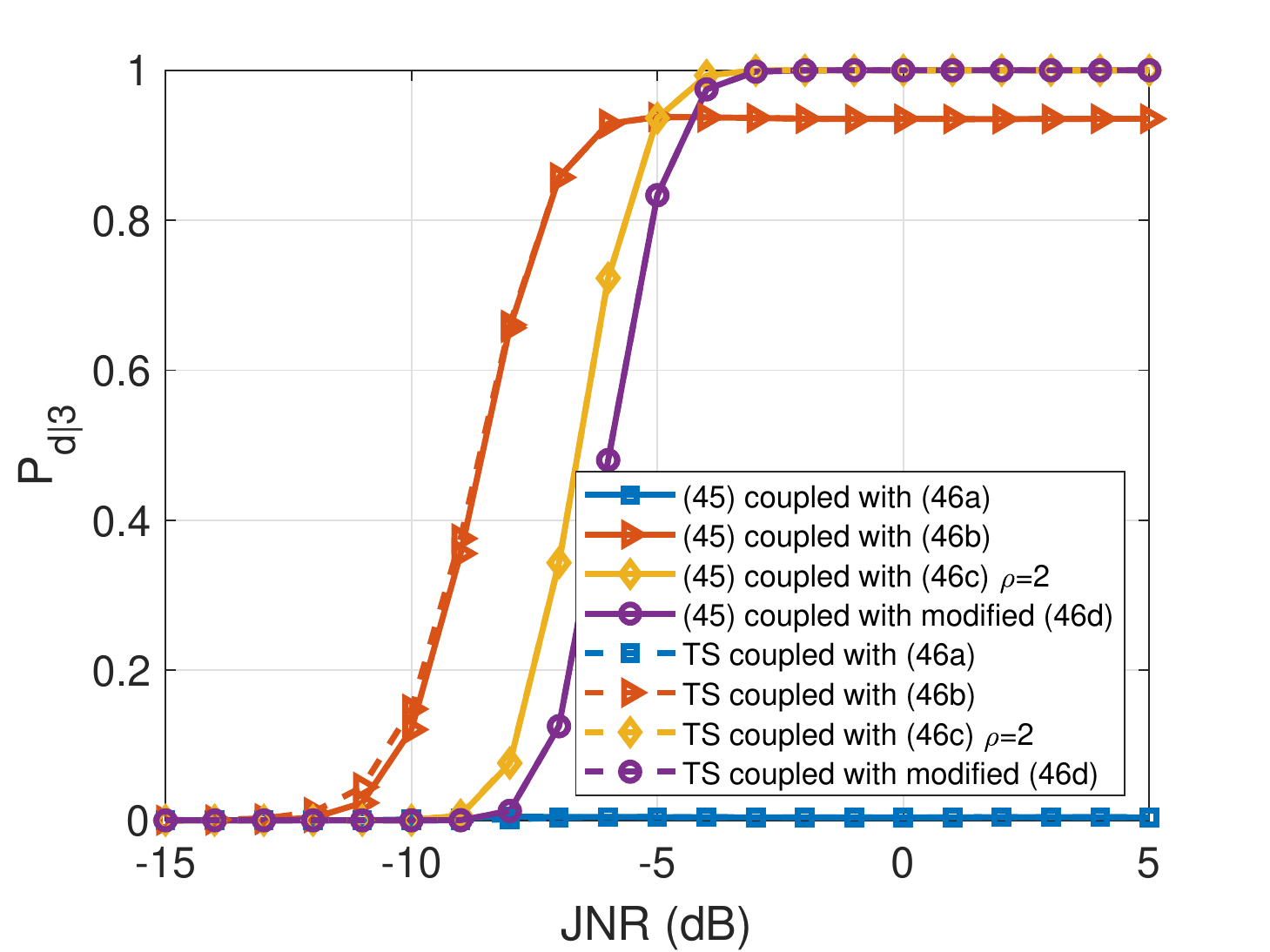}
\caption{Multiple Noise-Like Jammers Detection. $P_{d|3}$ versus JNR
for the considered architectures assuming $N=16$, $K=32$.}
\label{fignoise}
\end{figure}

\begin{figure}
\centering
\includegraphics[width=0.45\textwidth]{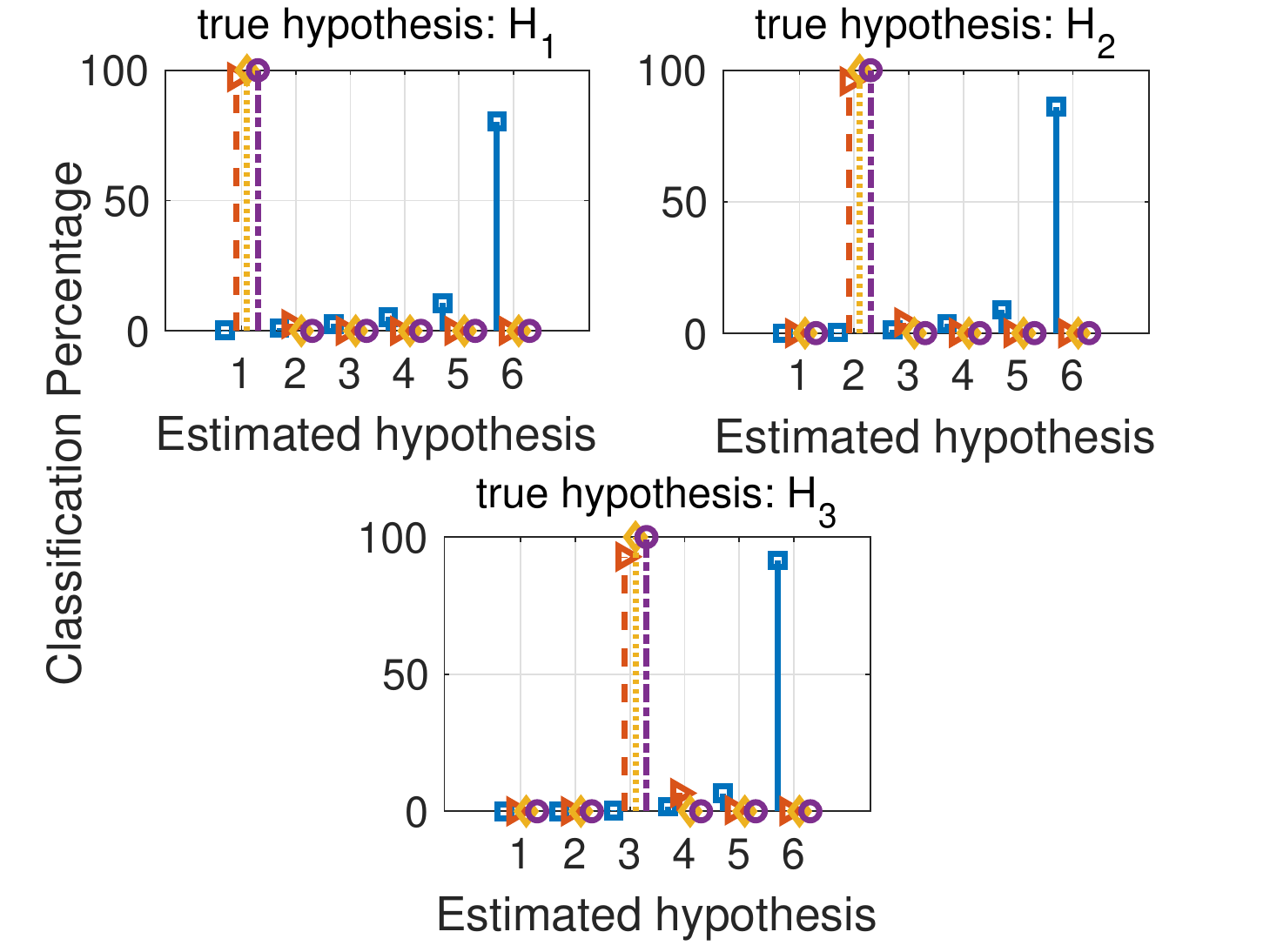}
\caption{Multiple Noise-Like Jammers Detection. Classification histograms for each hypothesis 
($\boxempty$-marked (\ref{eqn:decisionCriterion}) coupled with (\ref{eqn:penaltyMOSDet})(a), $\triangleright$-marked 
(\ref{eqn:decisionCriterion}) coupled with (\ref{eqn:penaltyMOSDet})(b), $\diamond$-marked (\ref{eqn:decisionCriterion}) coupled with (\ref{eqn:penaltyMOSDet})(c), $\rho=2$, 
$\circ$-marked (\ref{eqn:decisionCriterion}) coupled with modified (\ref{eqn:penaltyMOSDet})(d)) assuming $N=16$, $K=32$, and $\text{JNR=10}$ dB.}
\label{fignoisehist}
\end{figure}

\subsection{Detection in the presence of Coherent Jammers}
This second illustrative example concerns a scenario where the target echoes compete against a possible 
coherent jammer. Denoting by $\bz\in\C^{N\times 1}$ the cell under test and by $\bz_k$, $k=1,\ldots,K$, the training data set used to estimate the unknown Interference Covariance Matrix (ICM) \cite{BOR-Morgan}, 
the problem at hand can be recast as
\be
\left\{
\begin{array}{lll}
H_0: & \bz=\bn, & \! \bz_k=\bn_k, \ k = 1, \ldots, K, 
\\
H_{1}: & \bz=\bq + \bn, & \! \bz_k=\bn_k, \ k = 1, \ldots, K, 
\\
H_{2}: & \bz=\alpha\bv(\theta_T) + \bn, & \! \bz_k=\bn_k, \ k = 1, \ldots, K, 
\\
H_{3}: & \bz=\alpha\bv(\theta_T) + \bq + \bn, & \! \bz_k=\bn_k, \ k = 1, \ldots, K, 
\end{array}
\right.
\label{eqn:detectionProblemCoherentJammers}
\ee
where $\theta_T$ is the AOA of the target, $\bn,\bn_k\sim\cC\cN_N(\bzero,\bM)$, $k=1,\ldots,K$, are 
statistically independent, $\bM\in\C^{N\times N}$
is the unknown ICM, $\bq\in\C^{N\times 1}$ is the unknown spatial signature of the coherent jammer, and 
$\alpha\in\C$ is an unknown deterministic factor accounting for target and channel propagation effects.
Hereafter, for simplicity, we denote by $\bv$ the nominal steering vector of the target, namely $\bv(\theta_T)$.

In what follows, we exploit the subspace 
paradigm to model the coherent interference \cite{DD}, namely, $\bq=\bJ\ba$, where $\bJ\in\C^{N\times q}$ is 
a known full-column rank matrix whose columns span the jammer subspace\footnote{A priori information about $\bJ$ 
can be gathered by exploiting an Electronic Support Measure system that provides a rough estimate of 
the coherent jammer AOA.} and are linearly independent of $\bv$, while $\boa\in\C^{q\times 1}$
is the vector of the jammer coordinates. 

Now, using the results contained in \cite{kelly1986adaptive,DD,SD}, it is possible to show that
\begin{align}
&\widehat{\Lambda}_1(\bZ,1)= (K+1)\bigg\{\log(1+\bz^{\dag}\bS^{-1}\bz)\\ \nonumber
&-\log\left[1 + \bz^{\dag}\bS^{-1}\bz - \bz^{\dag}\bS^{-1}\bJ\left(\bJ^{\dag}\bS^{-1}\bJ\right)^{-1}\bJ^{\dag}\bS^{-1}\bz\right]\bigg\},
\\
&\widehat{\Lambda}_1(\bZ,2)= (K+1)\bigg[\log(1+\bz^{\dag}\bS^{-1}\bz)\\ \nonumber
&-\log\left(1 + \bz^{\dag}\bS^{-1}\bz - \frac{\bz^{\dag}\bS^{-1}\bv\bv^{\dag}\bS^{-1}\bz}{\bv^{\dag}\bS^{-1}\bv}\right)\bigg],
\\
&\widehat{\Lambda}_1(\bZ,3) =\\ \nonumber
& (K+1)\left[\log(1+\bz^{\dag}\bS^{-1}\bz)-\log\left(1 + \tilde{\bz}_S^{\dag}\tilde{\bP}_{\tilde{\bJ}_S}^{\perp}\tilde{\bz}_S\right)\right],
\end{align}
where $\bS=\sum_{k=1}^K\bz_k\bz_k^\dag$, $\tilde{\bP}_{\tilde{\bJ}_S}^{\perp} = \bI-\tilde{\bJ}_S\left(\tilde{\bJ}_S^{\dag}\tilde{\bJ}_S\right)^{-1}\tilde{\bJ}_S^{\dag}$, with $\tilde{\bJ}_S = \left(\bP_{\bS^{-1/2}\bv}^{\perp}\right)^{1/2}\bJ_S$, $\bJ_S = \bS^{-1/2}\bJ$, and
$$
\bP_{\bS^{-1/2}\bv}^{\perp} = \bI - \frac{\bS^{-1/2}\bv\bv^{\dag}\bS^{-1/2}}{\bv^{\dag}\bS^{-1}\bv^{\dag}},
$$
having set $\tilde{\bz}_S = \left(\bP_{\bS^{-1/2}\bv}^{\perp}\right)^{1/2}\bz_S$ and $\bz_S = \bS^{-1/2}\bz$.
The number of unknown parameters is $2q+N^2$, $2+N^2$, and $2+2q+N^2$ under $H_1$, $H_2$, and $H_3$, respectively
whereas $T=2(K+1)N$. Also in this case, \eqref{eqn:penaltyMOSDet}(d) is replaced by 
the asymptotic approximation of \eqref{eqn:BIC_full} for $K\rightarrow +\infty$ whose expression
is the same as in the previous subsection.

The probability of correct detection is estimated assuming that $H_3$ holds in a scenario where 
the actual AOAs of the coherent jammer and target are $40^\circ$ and $0^\circ$, respectively.
For simplicity, the subspace of the coherent interferer is set through the following matrix 
$\bJ=[\bv(\theta_{J,1}),\bv(\theta_{J,2}),\bv(\theta_{J,3})]$,
where $\theta_{J,1} = 35^{\circ}$, $\theta_{J,2} = 40^{\circ}$, and $\theta_{J,3} = 45^{\circ}$. 
The Jammer-to-Clutter plus Noise Ratio (JCNR) is defined as
$\text{JCNR} = \bv(\theta_{J,2})^\dag \bM^{-1} \bv(\theta_{J,2})$. 
Finally, we assume that the $(n,m)$th entry of the ICM is given by
\begin{equation}\label{eq:expcov}
\bM(n,m) = \sigma_n^2 + \sigma_c^2 \rho_c^{\left|n-m\right|},
\end{equation}
with $\rho_c=0.95$ the one-lag correlation coefficient and $\sigma^2_c$ the clutter power.

Figure \ref{figcoherent1} shows the $P_{d|3}$ versus JCNR assuming $\text{SNR} = 20$ dB and CNR$=20$ dB. 
The curves related to the TS counterparts, where the second stage is the GLRT corresponding
to the hypothesis selected by the first stage, are also reported. Also in this case, the curves
for the proposed architectures and their analogous TS counterparts are perfectly
overlapped. The floor observed at low JCNR values is due to the presence of strong useful signal components
which increase the decision statistic value. The figure draws a ranking where the first position
is occupied by (\ref{eqn:decisionCriterion}) coupled with \eqref{eqn:penaltyMOSDet}(a) followed by
(\ref{eqn:decisionCriterion}) coupled with \eqref{eqn:penaltyMOSDet}(b), (\ref{eqn:decisionCriterion}) coupled with \eqref{eqn:penaltyMOSDet}(c),
and, finally, (\ref{eqn:decisionCriterion}) coupled with modified \eqref{eqn:penaltyMOSDet}(d). However, 
such ranking may be misleading
since (\ref{eqn:decisionCriterion}) coupled with \eqref{eqn:penaltyMOSDet}(a) or \eqref{eqn:penaltyMOSDet}(b) is inclined to overestimate the hypothesis order. 
As a matter of fact, from the classification histograms
shown in Figure \ref{figcoherent2}, such behavior is evident, whereas 
(\ref{eqn:decisionCriterion}) coupled with \eqref{eqn:penaltyMOSDet}(c) or 
modified \eqref{eqn:penaltyMOSDet}(d) exhibits a more reliable performance
with a percentage of correct classification greater than $80$\% for each hypothesis.
\begin{figure}
\centering
\includegraphics[width=0.45\textwidth]{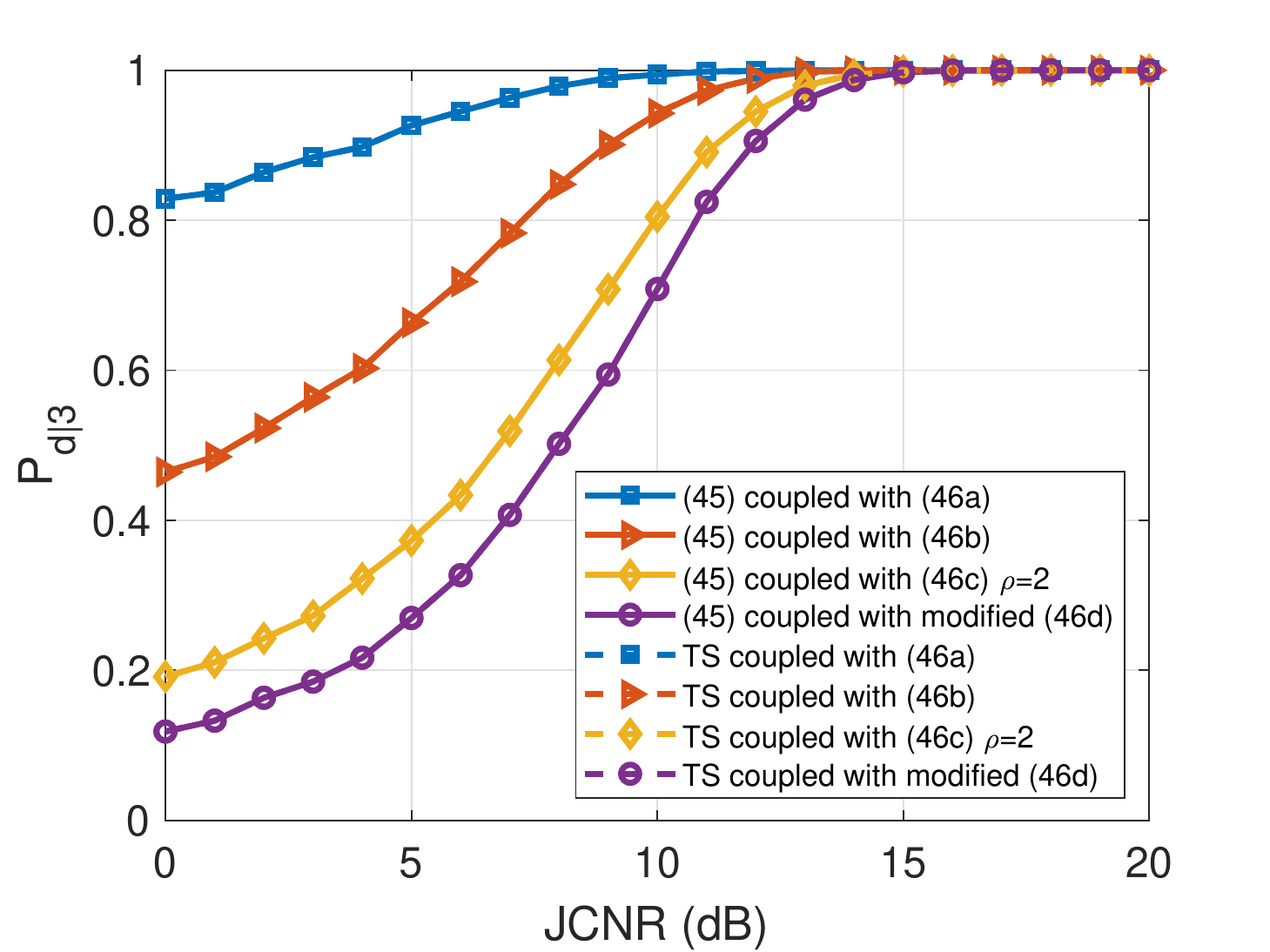}
\caption{Detection in the presence of coherent jammers. $P_{d|3}$ versus JCNR for 
the considered architectures assuming $N=16$, $K=32$, and $\text{SNR}=20$ dB.}
\label{figcoherent1}
\end{figure}
\begin{figure}
\centering
\includegraphics[width=0.99\columnwidth]{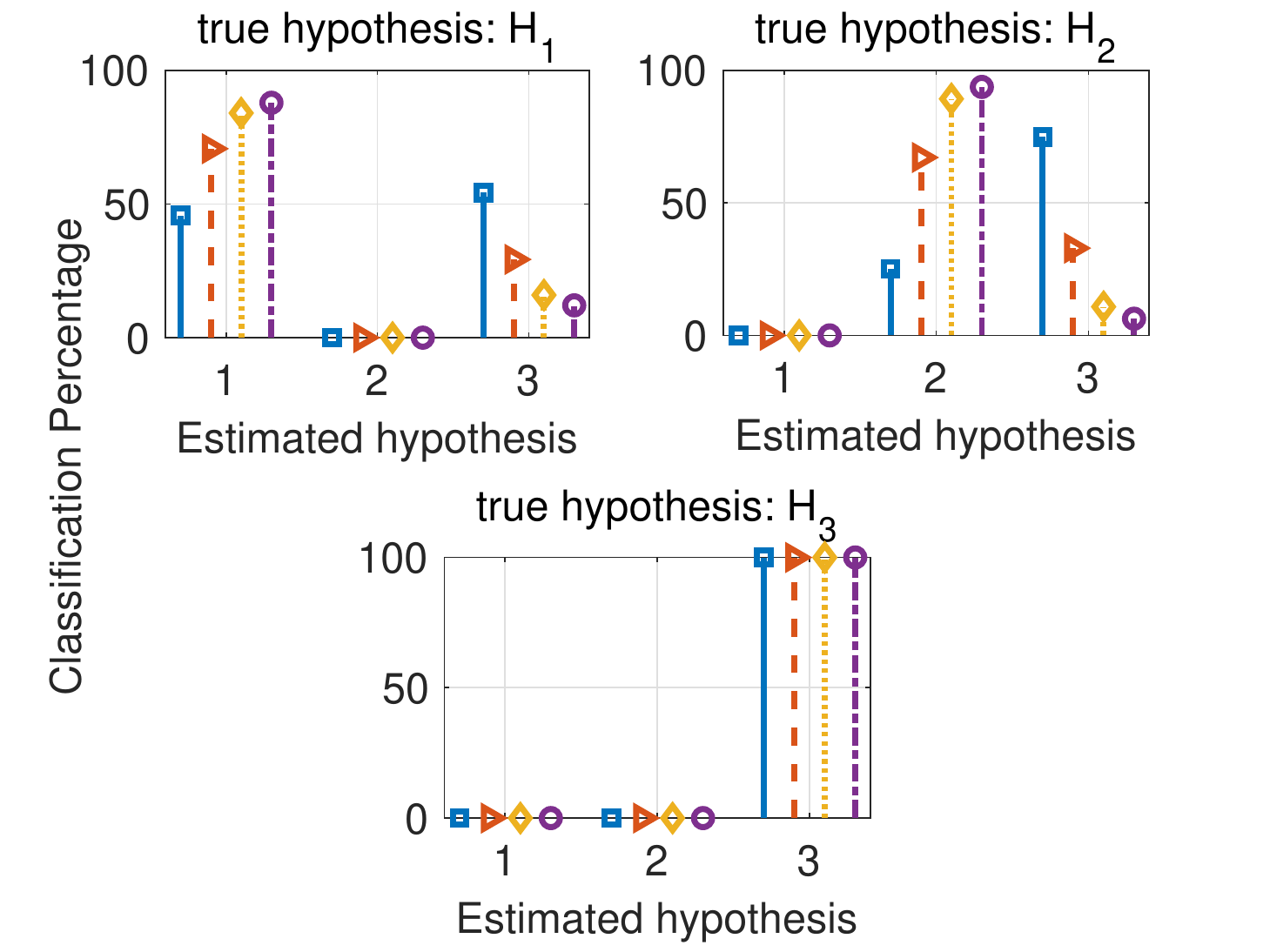}
\caption{Detection in the presence of coherent jammers. 
Classification histograms for each hypothesis 
($\boxempty$-marked (\ref{eqn:decisionCriterion}) coupled with (\ref{eqn:penaltyMOSDet})(a), $\triangleright$-marked 
(\ref{eqn:decisionCriterion}) coupled with (\ref{eqn:penaltyMOSDet})(b), 
$\diamond$-marked (\ref{eqn:decisionCriterion}) coupled 
with (\ref{eqn:penaltyMOSDet})(c) $\rho=2$, $\circ$-marked (\ref{eqn:decisionCriterion}) 
coupled with modified (\ref{eqn:penaltyMOSDet})(d)) assuming $N=16$, $K=32$, $\text{SNR}=20$ dB, 
and $\text{JCNR}=20$ dB.}\label{figcoherent2}
\end{figure}

\subsection{Range-spread Targets}
The last illustrative example is related to range-spread targets whose size is greater than the range resolution 
leading to scattering centers belonging to several contiguous range bins. Moreover, we assume that target
size in terms of contiguous range bins and position are not known. Thus, instead of the conventional detection
procedure which examines one range bin at time, we resort to the proposed framework to devise an 
architecture that processes a set of contiguous range bins (window under test).
This scenario naturally leads to a multiple-hypothesis test where each alternative hypothesis is 
associated with a specific target size and position. Note that this detection problem is similar to that 
addressed in \cite{1605248} except for the size of the extended target which is not known here. 
Summarizing, the multiple hypothesis test modeling the problem at hand has the following expression
\be
\label{eqn:ProblemRSTargets}
\left\{
\begin{array}{l}
H_0: 
\bz_l=\bn_l,\, l\in \Omega_T,\, \bz_k=\bn_k,\, k \in \Omega_S\\
H_{m}: 
\left\{
\begin{array}{l}
\bz_l = \alpha_l\bv+\bn_l,\, l\in \Omega_m\subseteq \Omega_T,\, \bz_k=\bn_k,\, k \in \Omega_S\\
\bz_l = \bn_l,\, l\in \Omega_T\setminus\Omega_m
\end{array}
\right.
\end{array}
\right.
\ee
where $\bv$ is defined in the previous subsection, 
$\Omega_T=\left\{1,\ldots, L\right\}$, $\Omega_S=\left\{L+1,\ldots, L+K\right\}$, 
$\Omega_m\subseteq \Omega_T$ is the set of consecutive range bins containing useful signal components.

Based upon the results of \cite{GLRT-based}, the logarithm of the GLR for the above problem can be written as
\be
\begin{split}
\widehat{\Lambda}_{1}&(\bZ;m) = (L+K)\bigg\{\log\det\left(\bS_0\right)-\log\det\bigg[\\
&\sum_{l\in\Omega_m} \left(\bz_l-\frac{\bv^{\dag}\bS_1^{-1}\bz_l}{\bv^{\dag}\bS_1^{-1}\bv}\bv\right) \left(\bz_l-\frac{\bv^{\dag}\bS_1^{-1}\bz_l}{\bv^{\dag}\bS_1^{-1}\bv}\bv\right)^{\dag} +\bS_1\bigg]\bigg\},
\end{split}
\ee
where $\bS_0 = \sum_{l=1}^{L} \bz_l\bz_l^{\dag} + \bS$ with $\bS = \sum_{k=L+1}^{L+K} \bz_k\bz_k^{\dag}$ and
$\bS_1 = \sum_{l\in\Omega_T \setminus\Omega_m} \bz_l\bz_l^{\dag} + \bS$.
The number of unknown parameters is $p=2 |\Omega_m|+1+N^2$ while $T=2(L+K)N$. 
The penalty factor raising from the asymptotic approximation of \eqref{eqn:BIC_full} for $K\rightarrow +\infty$
is the same as in the previous subsections.

The performance of the proposed MOS-based detectors are assessed assuming the same ICM as 
in the previous 
subsection (see \eqref{eq:expcov}). In addition, the target is assumed to be located 
at the antenna boresight and occupying the range bins $\{4,5\}$ within a window under test
of size $L=10$. After suitably labeling all the alternative hypotheses through the index
$m=1,\ldots,55$, the index associated with the actual position and size of the target is $m=14$.
Thus, the figure of merit becomes $P_{d|14}$ and it is estimated in Figure \ref{figspread} 
as a function of the SINR whose expression is given by
$\text{SINR} = {\sum_{l\in \{4,5\}}} |\alpha_l|^2 \bv^{\dag} \bM^{-1} \bv$ with CNR$=20$ dB.
\begin{figure}
\centering
\includegraphics[width=0.45\textwidth]{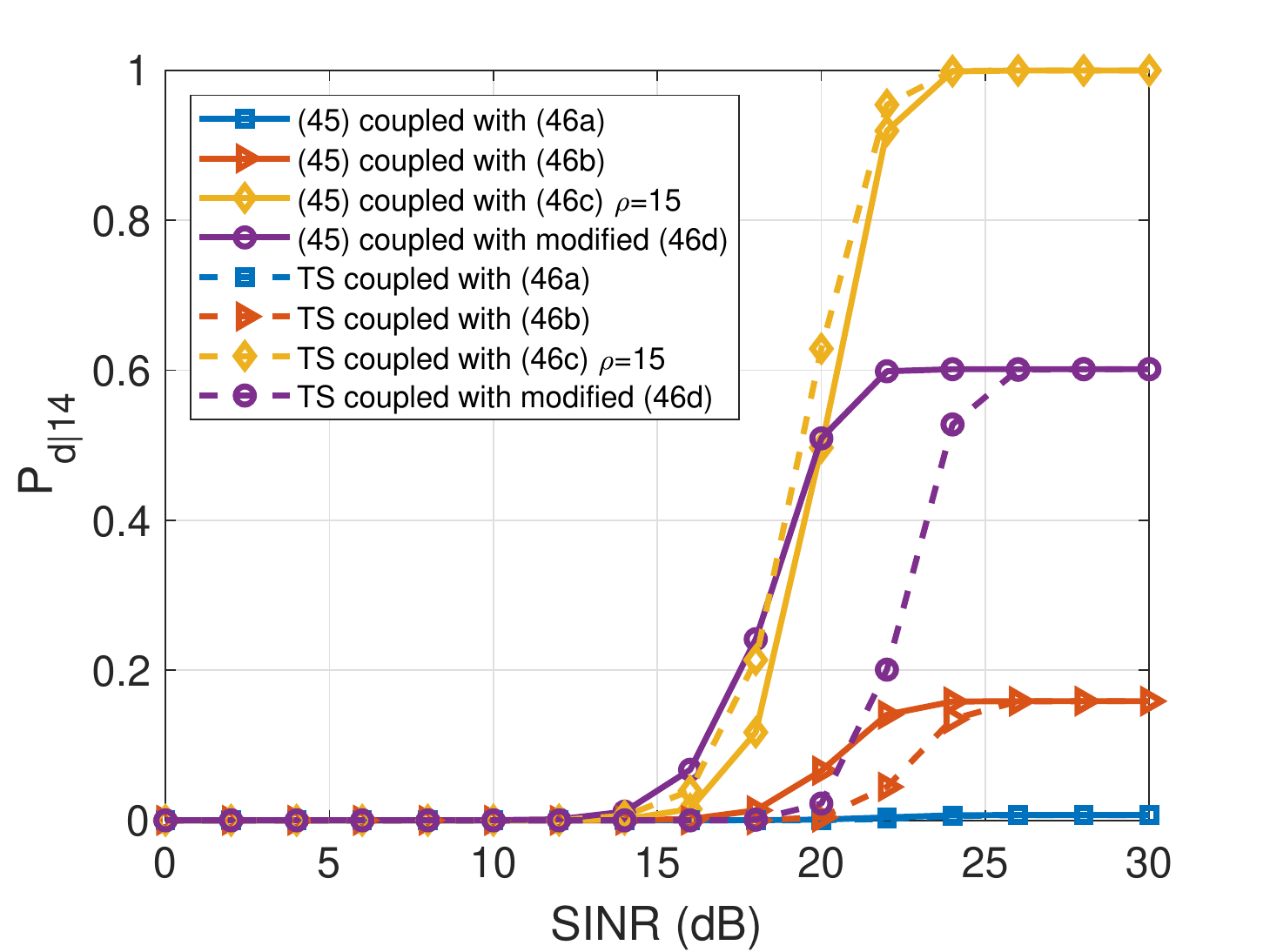}
\caption{Range-spread targets detection. $P_{d|14}$ versus SINR for a radar with $N=16$ spatial channels and $K=32$ secondary data.}\label{figspread}
\end{figure}
The results show that only the curve related to \eqref{eqn:decisionCriterion} coupled with \eqref{eqn:penaltyMOSDet}(c)
(and its TS counterpart) with $\rho=15$ achieves $P_{d|14}=1$, whereas the curves of the other decision rules 
saturate to low values with the newly proposed architectures overcoming the TS competitors. 
This situation can be explained analyzing the next figure where in place of the
classification histograms, we show the Root Mean Square Error (RMSE) for target size and position. The 
choice of these figures of merit
is due to the huge number of hypotheses that makes the histogram readability very difficult. 
The inspection of this figure points out that the estimates returned by
\eqref{eqn:decisionCriterion} coupled with \eqref{eqn:penaltyMOSDet}(a), \eqref{eqn:penaltyMOSDet}(b), 
and modified \eqref{eqn:penaltyMOSDet}(d) take on a constant value for high SINR values
with the side effect of lowering the correct detection
probability. On the other hand, the RMSE curves associated with
\eqref{eqn:decisionCriterion} coupled with \eqref{eqn:penaltyMOSDet}(c) 
decrease to zero as the SINR increases by virtue of the tuning parameter,
which is set in order to contrast the overestimation behavior generated by the presence of
several nested hypotheses.

\begin{figure}
\centering
\subfigure[]{\includegraphics[width=0.24\textwidth]{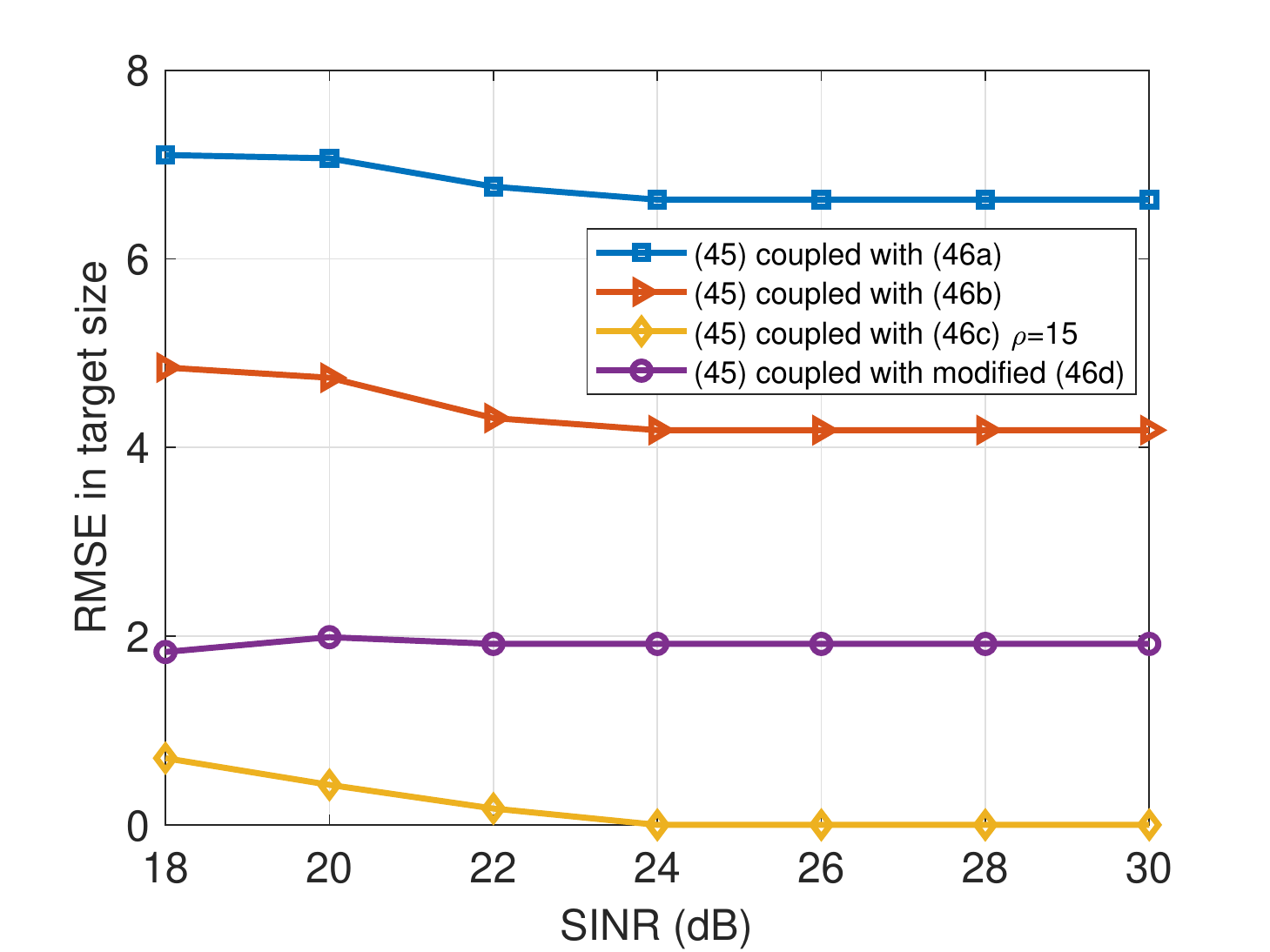}}
\subfigure[]{\includegraphics[width=0.24\textwidth]{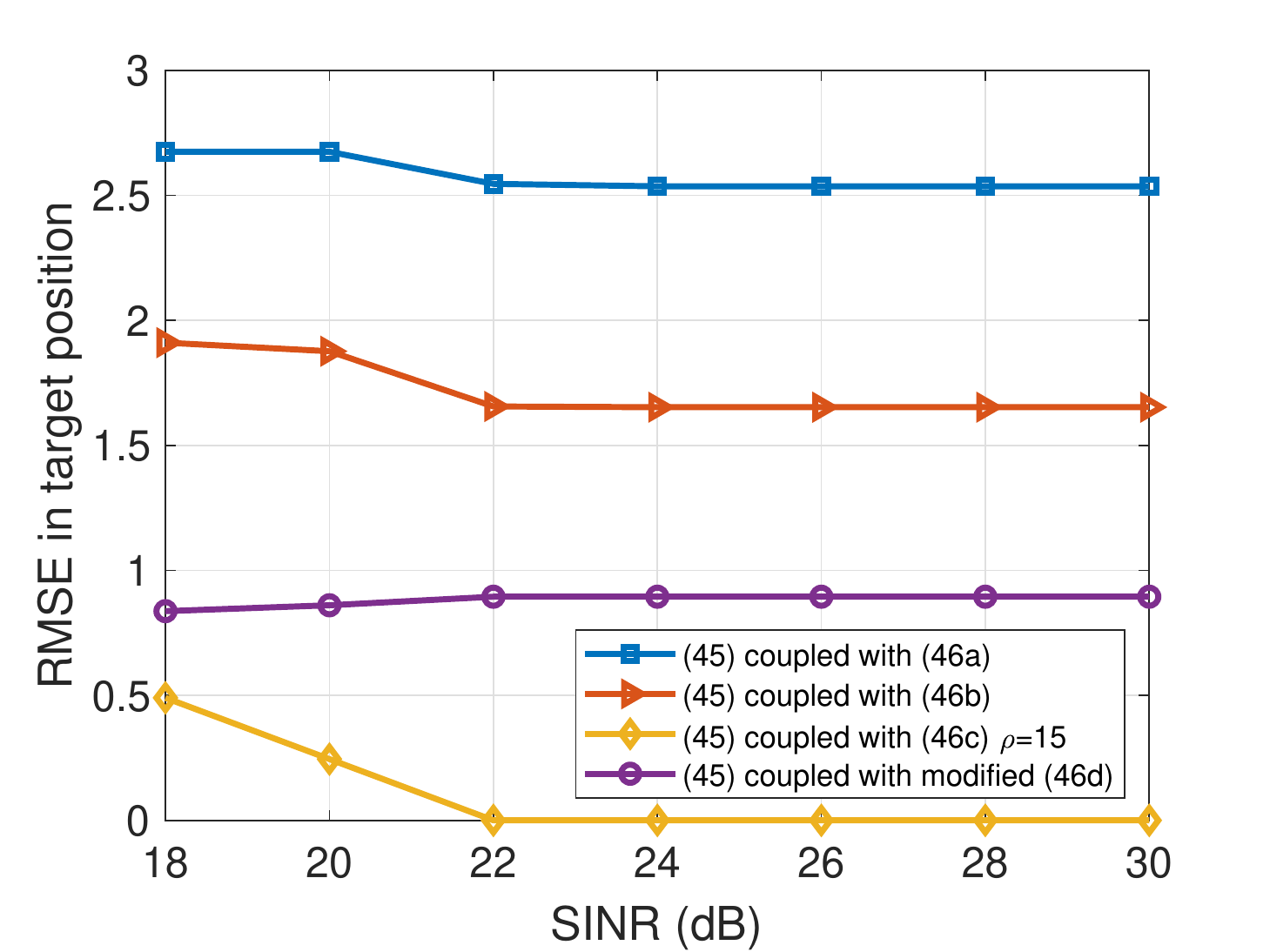}}
\caption{Range-spread targets detection. RMSE of estimated (a) extention and (b) position versus SINR. 
A radar with $N=16$ spatial channels and $K=32$ secondary data is considered.}\label{figspread2}
\end{figure}

\section{Conclusions}\label{Sec:Conclusions}
In this paper, we have developed a new framework based upon information theoretic criteria to address
multiple hypothesis testing problems. Specifically, we have considered multiple (possible nested) alternative 
hypotheses and only one null hypothesis. Such problems frequently rise in radar scenarios of practical
interest. The proposed design procedure exploits the KLIC to come up with decision statistics that 
incorporate a penalized GLR, where the penalty term depends on the number of unknown parameters, which, in turn,
is related to the specific alternative hypothesis under consideration. Interestingly,
such framework provides an information theoretic derivation of the GLRT and lays the foundation
for the design of detection architectures capable of operating in the presence
of multiple alternative hypotheses solving the limitations of the MLA which exhibits an overestimation
inclination when the hypotheses are nested. Moreover, we have shown that under some specific conditions,
decision schemes devised within this framework ensure the CFAR property.
Finally, we have applied the new design procedure to three different radar detection problems and investigated
the performance of four new detectors also in comparison with the analogous TS counterparts, whose structure 
is more complex to be implemented in real systems and that are slightly more time demanding than the former.
The analysis has singled out \eqref{eqn:decisionCriterion} coupled
with \eqref{eqn:penaltyMOSDet}(c) as the first choice at least for the considered cases, whereas
\eqref{eqn:decisionCriterion} coupled with \eqref{eqn:penaltyMOSDet}(d) arises as the second choice
since it is a parameter-free decision scheme and represents a good compromise between detection 
and classification performance in most cases.
However, it is important to notice that different study cases require further investigations 
which may lead to different results.

Future research tracks may encompass the analysis where the GLR is replaced by asymptotically equivalent
statistics or the derivation of other decision rules based upon the joint ML and maximum a posteriori 
estimation procedure testing several priors for the number of parameters.

\section*{Acknowledgments}
The authors would like to thank Prof. A. De Maio and Prof. G. Ricci for the interesting discussions.

\appendix
Decision rule \eqref{eqn:decisionCriterion} can be viewed as the result of a suitable log-likelihood regularization, which aims at overcoming the previously described limitations of the ML approach in the 
case of nested models. To this end, we assume that the number of parameters of interest under $H_1$ is a discrete random variable with a prior promoting low-dimensional model in order to mitigate the natural 
inclination of the ML approach to overestimate the model size.

With the above remarks in mind, a possible probability mass function for $p_{r}$ is chosen as
\be
\pi(p_{r}) = \frac{1}{A} e^{-g(p_{r})}, \quad p_{r}\in\Omega_r,
\ee
where $g(\cdot)$ is a positive and increasing function of $p_r$ and $A$ is a normalization constant such that
$A = \sum_{p_r\in\Omega_r} e^{-g(p_{r})}$.
For the next development, it is important to note that if $p_r$ is a continuous random variable, the joint 
pdf of $\bZ$ and $p_r$ can be recast as
\be
f(\bZ,p_r;\btheta_{p_r}, \btheta_{p_s}) = f(\bZ;\btheta_{r,1}, \btheta_{s}|p_r)f(p_r),
\ee
where $f(p_r)$ is the pdf of $p_r$. However, since $p_r$ is a discrete random variable, the joint pdf of 
$\bZ$ and $p_r$ is available only in the generalized sense exhibiting Dirac delta functions in correspondence 
of the values assumed by $p_r$. Thus, if we consider the following decision rule
\be
\frac{ \dmax_{p_r\in\Omega_r} \dmax_{\btheta_{r,1}, \btheta_{s}} f(\bZ; \btheta_{r,1}, \btheta_{s}|p_r)\pi(p_r)} {\dmax_{\btheta_{r,0}, \btheta_{s}} f_0(\bZ;\btheta_{r,0}, \btheta_{s})}\test\eta,
\label{eqn:MAP_ML_test_00}
\ee
at the numerator, we attempt to maximize the multipliers of the Dirac delta functions, i.e., it only 
focuses on the lines where the probability masses are located. In this sense, the optimization at
the numerator can be interpreted in terms of the joint ML and maximum a posteriori estimation procedure \cite{jointMAP_ML,Stoica1}. The above architecture can be written as
\be
\label{eqn:MAP_ML_test}
\begin{split}
\dmax_{p_r\in\Omega_r}&\dmax_{\btheta_{r,1}, \btheta_{s}} \log{ f(\bZ;\btheta_{r,1}, \btheta_{s}|p_r)
\pi(p_r)}\\
& - \dmax_{\btheta_{r,0}, \btheta_{s}} \log f_0(\bZ;\btheta_{r,0}, \btheta_{s})\test\eta,
\end{split}
\ee
After maximizing with respect to $\btheta_{r,1}$, $\btheta_{r,0}$, and $\btheta_{s}$, 
given $p_r$, \eqref{eqn:MAP_ML_test} can be recast as
\be
\dmax_{p_r\in\Omega_r} \left\{\log\frac{f(\bZ;\widehat{\btheta}_{r,1}, \widehat{\btheta}_{s,1}|p_r)}
{f(\bZ;\widehat{\btheta}_{r,0},\widehat{\btheta}_{s,0})} +\log \pi(p_r) \right\}\test \eta.
\ee
The above decision scheme has the same structure as \eqref{eqn:decisionCriterion} and, when $g(p_r)=h(p_r)$, they
are equivalent. Nevertheless, several alternatives can be used for $g(p_r)$ to come up with different
penalization terms.

%
\bibliographystyle{IEEEtran}
\bibliography{group_bib}

\begin{thebibliography}{10}
\providecommand{\url}[1]{#1}
\csname url@samestyle\endcsname
\providecommand{\newblock}{\relax}
\providecommand{\bibinfo}[2]{#2}
\providecommand{\BIBentrySTDinterwordspacing}{\spaceskip=0pt\relax}
\providecommand{\BIBentryALTinterwordstretchfactor}{4}
\providecommand{\BIBentryALTinterwordspacing}{\spaceskip=\fontdimen2\font plus
\BIBentryALTinterwordstretchfactor\fontdimen3\font minus
  \fontdimen4\font\relax}
\providecommand{\BIBforeignlanguage}[2]{{%
\expandafter\ifx\csname l@#1\endcsname\relax
\typeout{** WARNING: IEEEtran.bst: No hyphenation pattern has been}%
\typeout{** loaded for the language `#1'. Using the pattern for}%
\typeout{** the default language instead.}%
\else
\language=\csname l@#1\endcsname
\fi
#2}}
\providecommand{\BIBdecl}{\relax}
\BIBdecl

\bibitem{bergin2002gmti}
J.~S. Bergin, P.~M. Techau, W.~L. Melvin, and J.~R. Guerci, ``{GMTI STAP in
  Target-Rich Environments: Site-Specific Analysis},'' in \emph{Radar
  Conference, 2002. Proceedings of the IEEE}.\hskip 1em plus 0.5em minus
  0.4em\relax IEEE, 2002, pp. 391--396.

\bibitem{wicksOutliers}
M.~C. {Wicks}, W.~L. {Melvin}, and P.~{Chen}, ``{An Efficient Architecture for
  Nonhomogeneity Detection in Space-Time Adaptive Processing Airborne Early
  Warning Radar},'' in \emph{Radar 97 (Conf. Publ. No. 449)}, October 1997, pp.
  295--299.

\bibitem{AdveOuliers}
R.~S. {Adve}, T.~B. {Hale}, and M.~C. {Wicks}, ``{Transform Domain Localized
  Processing using Measured Steering Vectors and Non-Homogeneity Detection},''
  in \emph{Proceedings of the 1999 IEEE Radar Conference. Radar into the Next
  Millennium (Cat. No.99CH36249)}, April 1999, pp. 285--290.

\bibitem{JiangOutliers}
L.~{Jiang} and T.~{Wang}, ``{Robust Non-Homogeneity Detector based on
  Reweighted Adaptive Power Residue},'' \emph{IET Radar, Sonar Navigation},
  vol.~10, no.~8, pp. 1367--1374, 2016.

\bibitem{RangaswamyOutliers}
M.~{Rangaswamy}, ``{Non-Homogeneity Detector for Gaussian and non-Gaussian
  Interference Scenarios},'' in \emph{Sensor Array and Multichannel Signal
  Processing Workshop Proceedings, 2002}, Aug 2002, pp. 528--532.

\bibitem{ScheerMelvin}
W.~L. Melvin and J.~A. Scheer, \emph{Principles of Modern Radar: Advanced
  Techniques}, S.~Publishing, Ed., Edison, NJ, 2013.

\bibitem{PiaIET}
P.~{Addabbo}, A.~{Aubry}, A.~{De Maio}, L.~{Pallotta}, and S.~L. {Ullo}, ``{HRR
  Profile Estimation using SLIM},'' \emph{IET Radar, Sonar Navigation},
  vol.~13, no.~4, pp. 512--521, 2019.

\bibitem{WLiuRao}
W.~Liu, W.~Xie, and Y.~Wang, ``{Rao and Wald Tests for Distributed Targets
  Detection With Unknown Signal Steering},'' \emph{IEEE Signal Processing
  Letters}, vol.~20, no.~11, pp. 1086--1089, 2013.

\bibitem{JunLiu02}
J.~Liu, H.~Li, and B.~Himed, ``{Persymmetric Adaptive Target Detection with
  Distributed MIMO Radar},'' \emph{IEEE Transactions on Aerospace and
  Electronic Systems}, vol.~51, no.~1, pp. 372--382, 2015.

\bibitem{Gerlach}
K.~Gerlach and M.~J. Steiner, ``Adaptive detection of range distributed
  targets,'' \emph{IEEE Transactions on Signal Processing}, vol.~47, no.~7, pp.
  1844--1851, July 1999.

\bibitem{DD}
F.~Bandiera, O.~Besson, D.~Orlando, G.~Ricci, and L.~L. Scharf, ``{GLRT-Based
  Direction Detectors in Homogeneous Noise and Subspace Interference},''
  \emph{IEEE Transactions on Signal Processing}, vol.~55, no.~6, pp.
  2386--2394, June 2007.

\bibitem{GLRT-based}
E.~Conte, A.~De~Maio, and G.~Ricci, ``{GLRT-Based Adaptive Detection Algorithms
  for Range-Spread Targets},'' \emph{IEEE Transactions on Signal Processing},
  vol.~49, no.~7, pp. 1336--1348, July 2001.

\bibitem{BR_multipleTargets}
F.~{Bandiera} and G.~{Ricci}, ``{Adaptive Detection and Interference Rejection
  of Multiple Point-Like Radar Targets},'' \emph{IEEE Transactions on Signal
  Processing}, vol.~54, no.~12, pp. 4510--4518, December 2006.

\bibitem{BOR_multipleTargets}
F.~{Bandiera}, D.~{Orlando}, and G.~{Ricci}, ``{CFAR Detection of Extended and
  Multiple Point-Like Targets Without Assignment of Secondary Data},''
  \emph{IEEE Signal Processing Letters}, vol.~13, no.~4, pp. 240--243, April
  2006.

\bibitem{EW101}
D.~Adamy, \emph{EW101: A First Course in Electronic Warfare}, A.~House, Ed.,
  Norwood, MA, 2001.

\bibitem{FarinaSkolnik}
A.~Farina, ``{ECCM Techniques},'' in \emph{Radar Handbook}, M.~I. Skolnik,
  Ed.\hskip 1em plus 0.5em minus 0.4em\relax McGraw-Hill, 2008, ch.~24.

\bibitem{FarinaECM}
F.~{Bandiera}, A.~{Farina}, D.~{Orlando}, and G.~{Ricci}, ``{Detection
  Algorithms to Discriminate Between Radar Targets and ECM Signals},''
  \emph{IEEE Transactions on Signal Processing}, vol.~58, no.~12, pp.
  5984--5993, December 2010.

\bibitem{carotenutoNLJ}
V.~{Carotenuto}, C.~{Hao}, D.~{Orlando}, A.~{De Maio}, and S.~{Iommelli},
  ``{Detection of Multiple Noise-like Jammers for Radar Applications},'' in
  \emph{2018 5th IEEE International Workshop on Metrology for AeroSpace
  (MetroAeroSpace)}, June 2018, pp. 328--333.

\bibitem{GeneralizedNP}
J.~{Stuller}, ``{Generalized Likelihood Signal Resolution},'' \emph{IEEE
  Transactions on Information Theory}, vol.~21, no.~3, pp. 276--282, May 1975.

\bibitem{Anderson}
K.~P. Burnham and D.~R. Anderson, \emph{{Model Selection And Multimodel
  Inference, A Practical Information-Theoretic Approach}}, 2nd~ed.\hskip 1em
  plus 0.5em minus 0.4em\relax New York, USA: Springer-Verlag, 2002.

\bibitem{Stoica1}
P.~Stoica and Y.~Selen, ``{Model-Order Selection: A Review of Information
  Criterion Rules},'' \emph{IEEE Signal Processing Magazine}, vol.~21, no.~4,
  pp. 36--47, 2004.

\bibitem{BICneath}
A.~A. Neath and J.~E. Cavanaugh, ``{The Bayesian Information Criterion:
  Background, Derivation, and Applications},'' \emph{WIREs Computational
  Statistics}, vol.~4, no.~2, pp. 199--203, March 2012.

\bibitem{BhansaliGIC}
R.~J. Bhansali and D.~Y. Downham, ``{Some Properties of the Order of an
  Autoregressive Model Selected by a Generalization of Akaike's FPE
  Criterion},'' \emph{Biometrika}, vol.~64, pp. 547--551, 1977.

\bibitem{kailathMultipleHypotheses}
M.~{Wax} and T.~{Kailath}, ``{Detection of Signals by Information Theoretic
  Criteria},'' \emph{IEEE Transactions on Acoustics, Speech, and Signal
  Processing}, vol.~33, no.~2, pp. 387--392, April 1985.

\bibitem{FishlerMultHyp}
E.~{Fishler}, M.~{Grosmann}, and H.~{Messer}, ``{Detection of Signals by
  Information Theoretic Criteria: General Asymptotic Performance Analysis},''
  \emph{IEEE Transactions on Signal Processing}, vol.~50, no.~5, pp.
  1027--1036, May 2002.

\bibitem{VanTrees4}
H.~L. Van~Trees, \emph{{Optimum Array Processing (Detection, Estimation, and
  Modulation Theory, Part IV)}}.\hskip 1em plus 0.5em minus 0.4em\relax John
  Wiley \& Sons, 2002.

\bibitem{kullback1951}
S.~Kullback and R.~A. Leibler, ``{On Information and Sufficiency},'' \emph{Ann.
  Math. Statist.}, vol.~22, no.~1, pp. 79--86, 03 1951.

\bibitem{jointMAP_ML}
A.~Yeredor, ``{The Joint MAP-ML Criterion and its Relation to ML and to
  Extended Least-Squares},'' \emph{IEEE Transactions on Signal Processing},
  vol.~48, no.~12, pp. 3484--3492, 2000.

\bibitem{Richards}
M.~A. Richards, J.~A. Scheer, and W.~A. Holm, \emph{Principles of Modern Radar:
  Basic Principles}, S.~Publishing, Ed., Raleigh, NC, 2010.

\bibitem{KayBook}
S.~M. Kay, \emph{{Fundamentals of Statistical Signal Processing: Detection
  Theory}}, P.~Hall, Ed., 1998, vol.~2.

\bibitem{1605248}
F.~{Bandiera}, D.~{Orlando}, and G.~{Ricci}, ``{CFAR Detection of Extended and
  Multiple Point-Like Targets Without Assignment of Secondary Data},''
  \emph{IEEE Signal Processing Letters}, vol.~13, no.~4, pp. 240--243, April
  2006.

\bibitem{8781902}
L.~{Yan}, P.~{Addabbo}, C.~{Hao}, D.~{Orlando}, and A.~{Farina}, ``{New ECCM
  Techniques Against Noise-like and/or Coherent Interferers},'' \emph{IEEE
  Transactions on Aerospace and Electronic Systems}, pp. 1--1, 2019.

\bibitem{8835670}
L.~{Yan}, C.~{Hao}, P.~{Addabbo}, D.~{Orlando}, and A.~{Farina}, ``{An Improved
  Adaptive Radar Detector based on Two Sets of Training Data},'' in \emph{2019
  IEEE Radar Conference (RadarConf)}, April 2019, pp. 1--6.

\bibitem{senSinger}
P.~K. Sen and J.~M. Singer, \emph{{Large Sample Methods in Statistics: An
  Introduction with Applications}}.\hskip 1em plus 0.5em minus 0.4em\relax
  Springer US, 1993.

\bibitem{selen2007model}
Y.~Sel{\'e}n, \emph{{Model Selection and Sparse Modeling}}.\hskip 1em plus
  0.5em minus 0.4em\relax Department of Information Technology, Uppsala
  University, 2007.

\bibitem{cover2012elements}
T.~Cover and J.~Thomas, \emph{Elements of Information Theory}.\hskip 1em plus
  0.5em minus 0.4em\relax Wiley, 2012.

\bibitem{EGUCHI20062034}
``Interpreting kullback–leibler divergence with the neyman–pearson lemma,''
  \emph{Journal of Multivariate Analysis}, vol.~97, no.~9, pp. 2034 -- 2040,
  2006.

\bibitem{lehmann1986testing}
E.~L. Lehmann, J.~P. Romano, and G.~Casella, ``{Testing Statistical
  Hypotheses},'' vol. 150, 1986.

\bibitem{KayBook_Estimation}
S.~M. Kay, \emph{{Fundamentals of Statistical Signal Processing: Estimation
  Theory}}, P.~Hall, Ed., 1993, vol.~1.

\bibitem{geyer}
C.~J. Geyer, ``{Stat 8112 Lecture Notes: The Wald Consistency Theorem},''
  \emph{University of Minnesota, School of Statistics}, pp. 1--13, 2012.

\bibitem{robey1992cfar}
F.~C. Robey, D.~R. Fuhrmann, E.~J. Kelly, and R.~Nitzberg, ``{A CFAR Adaptive
  Matched Filter Detector},'' \emph{IEEE Trans. on Aerospace and Electronic
  Systems}, vol.~28, no.~1, pp. 208--216, 1992.

\bibitem{Yuri01}
Y.~I. Abramovich and B.~A. Johnson, ``{GLRT-Based Detection-Estimation for
  Undersampled Training Conditions},'' \emph{IEEE Transactions on Signal
  Processing}, vol.~56, no.~8, pp. 3600--3612, 2008.

\bibitem{HaoSP_HE}
C.~Hao, D.~Orlando, G.~Foglia, and G.~Giunta, ``Knowledge-based adaptive
  detection: Joint exploitation of clutter and system symmetry properties,''
  \emph{IEEE Signal Processing Letters}, vol.~23, no.~10, pp. 1489--1493,
  October 2016.

\bibitem{BOR-Morgan}
F.~Bandiera, D.~Orlando, and G.~Ricci, \emph{{Advanced Radar Detection Schemes
  Under Mismatched Signal Models}}, M.~. C.~P. Synthesis Lectures~on Signal
  Processing No.~8, Ed., San Rafael, US, 2009.

\bibitem{farina1986review}
A.~Farina and F.~A. Studer, ``{A Review of CFAR Detection Techniques in Radar
  Systems},'' \emph{Microwave Journal}, vol.~29, p. 115, 1986.

\bibitem{rohling1983radar}
H.~Rohling, ``{Radar CFAR Thresholding in Clutter and Multiple Target
  Situations},'' \emph{IEEE Trans. on Aerospace and Electronic Systems}, no.~4,
  pp. 608--621, 1983.

\bibitem{barkat1989cfar}
M.~Barkat, S.~D. Himonas, and P.~K. Varshney, ``{CFAR Detection for Multiple
  Target Situations},'' in \emph{IEE Proceedings F (Radar and Signal
  Processing)}, vol. 136, no.~5.\hskip 1em plus 0.5em minus 0.4em\relax IET,
  1989, pp. 193--209.

\bibitem{conte2002cfar}
E.~Conte, A.~{De Maio}, and G.~Ricci, ``{CFAR Detection of Distributed Targets
  in non-Gaussian Disturbance},'' \emph{IEEE Trans. on Aerospace and Electronic
  Systems}, vol.~38, no.~2, pp. 612--621, 2002.

\bibitem{roman2000parametric}
J.~R. Roman, M.~Rangaswamy, D.~W. Davis, Q.~Zhang, B.~Himed, and J.~H. Michels,
  ``{Parametric Adaptive Matched Filter for Airborne Radar Applications},''
  \emph{IEEE Trans. on Aerospace and Electronic Systems}, vol.~36, no.~2, pp.
  677--692, 2000.

\bibitem{gini2002covariance}
F.~Gini and M.~Greco, ``{Covariance Matrix Estimation for CFAR Detection in
  Correlated Heavy Tailed Clutter},'' \emph{Signal Processing}, vol.~82,
  no.~12, pp. 1847--1859, 2002.

\bibitem{kelly1986adaptive}
E.~J. Kelly, ``{An Adaptive Detection Algorithm},'' \emph{IEEE Transactions on
  Aerospace and Electronic Systems}, no.~2, pp. 115--127, 1986.

\bibitem{LehmannBook}
E.~L. Lehmann, \emph{{Testing Statistical Hypotheses}}, 2nd~ed.\hskip 1em plus
  0.5em minus 0.4em\relax New York, USA: Springer-Verlag, 1986.

\bibitem{scharf1991statistical}
L.~L. Scharf and C.~Demeure, \emph{{Statistical Signal Processing: Detection,
  Estimation, and Time Series Analysis}}, ser. Addison-Wesley Series in
  Electrical and Computer Engineering.\hskip 1em plus 0.5em minus 0.4em\relax
  Addison-Wesley Publishing Company, 1991.

\bibitem{eaton1983multivariate}
M.~Eaton, \emph{{Multivariate statistics: a vector space approach}}, ser.
  Lecture notes-monograph series.\hskip 1em plus 0.5em minus 0.4em\relax
  Institute of Mathematical Statistics, 1983.

\bibitem{robinson1996course}
D.~Robinson, S.~Axler, F.~Gehring, and P.~Halmos, \emph{A Course in the Theory
  of Groups}, ser. Graduate Texts in Mathematics.\hskip 1em plus 0.5em minus
  0.4em\relax Springer New York, 1996.

\bibitem{DoppioTraining}
V.~Carotenuto, A.~{De Maio}, D.~Orlando, and L.~Pallotta, ``{Adaptive Radar
  Detection Using Two Sets of Training Data},'' \emph{IEEE Transactions on
  Signal Processing}, vol.~66, no.~7, pp. 1791--1801, 2017.

\bibitem{SD}
F.~Bandiera, A.~De~Maio, A.~S. Greco, and G.~Ricci, ``Adaptive radar detection
  of distributed targets in homogeneous and partially homogeneous noise plus
  subspace interference,'' \emph{IEEE Transactions on Signal Processing},
  vol.~55, no.~4, pp. 1223--1237, April 2007.

\end{thebibliography}
\end{document}